\documentclass[11pt]{article}
\usepackage[utf8]{inputenc}
\usepackage{amsmath}
\usepackage{amsthm}
\usepackage{amssymb}
\usepackage[margin = 1in]{geometry}
\usepackage{todonotes}
\usepackage{multicol}
\usepackage{algorithm}
\usepackage[noend]{algpseudocode}
\usepackage[shortlabels]{enumitem}
\usepackage{comment}
\usepackage{lineno}
\usepackage{tabularx}
\usepackage{makecell}
\usepackage{hyperref}
\usepackage{graphicx}
\usepackage[caption=false]{subfig}
\usepackage{color}

\usepackage{authblk}

\newtheorem{theorem}{Theorem}
\newtheorem{lemma}[theorem]{Lemma}
\newtheorem{proposition}[theorem]{Proposition}
\newtheorem{definition}[theorem]{Definition}
\newtheorem{claim}{Claim}
\newtheorem{corollary}[theorem]{Corollary}
\newtheorem{observation}[theorem]{Observation}
\newtheorem{invariant}[theorem]{Invariant}

\newtheorem{fact}[theorem]{Fact}

\newcommand{\paren}[1]{\mathopen{}\left(#1\right)\mathclose{}}
\newcommand{\poly}{\operatorname{\text{{\rm poly}}}}

\title{The Power of Multi-Step Vizing Chains}
\author[1]{Aleksander B. G. Christiansen\thanks{This work was supported by VILLUM FONDEN grant 37507 ``Efficient Recomputations for Changeful Problems''.} \\
Technical University of Denmark}
\date{}%September 2022}

\begin{document}
%\linenumbers
\maketitle

\begin{abstract}
  Recent papers~\cite{BERNSHTEYN,duan,grebik2020measurable} have addressed different variants of the $(\Delta + 1)$-edge-colouring problem by concatenating or gluing together many Vizing chains to form what Bernshteyn~\cite{BERNSHTEYN} coined \emph{multi-step Vizing chains}. 
  In this paper, we consider the most general definition of this term and apply different multi-step Vizing chain constructions to prove combinatorial properties of edge-colourings that lead to (improved) algorithms for computing edge-colouring across different models of computation. 
  This approach seems especially powerful for constructing augmenting subgraphs which respect some notion of locality. 
  
  First, we construct strictly local multi-step Vizing chains and use them to show a local version of Vizing's Theorem thus confirming a recent conjecture of Bonamy, Delcourt, Lang and Postle~\cite{bonamy2020edge}. That is, we show that there exists a proper edge-colouring of a graph such that every edge $uv$ receives a colour from the list $\{1,2, \dots, \max\{d(u),d(v)\}+1\}$.
  Our proof is constructive and also implies an $O(n^2 \Delta)$ time algorithm for computing such a colouring.

  Then, we show that for any uncoloured edge there exists an augmenting subgraph of size $O(\Delta^{7}\log n)$, answering an open problem of Bernshteyn~\cite{BERNSHTEYN}. Chang, He, Li, Pettie and Uitto~\cite{pettie} show a lower bound of $\Omega(\Delta \log \frac{n}{\Delta})$ for the size of augmenting subgraphs, so the upper bound is asymptotically tight up to $\Delta$ factors.
  These ideas also extend to give a faster deterministic LOCAL algorithm for $(\Delta + 1)$-edge-colouring running in $\tilde{O}(\poly(\Delta)\log^6 n)$ rounds. 
  These results improve the dependency on $\log n$ compared to the recent breakthrough result of Bernshteyn~\cite{BERNSHTEYN}, who showed the existence of augmenting subgraphs of size $O(\Delta^6\log^2 n)$, and used these to give the first $(\Delta + 1)$-edge-colouring algorithm in the LOCAL model running in $O(\poly(\Delta, \log n))$ rounds.
  
  Finally for dynamic graphs, we show how to maintain a$(1+\varepsilon)\Delta$-edge-colouring fully adaptive to $\Delta$ in $O(\varepsilon^{-6} \log^9 n \log^6 \Delta)$ worst-case update time w.h.p without any restrictions on $\Delta$. 
  This should be compared to the edge-colouring algorithm of Duan, He and Zhang~\cite{duan} that runs in $O(\varepsilon^{-4}\log^8 n)$ amortised update time w.h.p under the condition that $\Delta = \Omega(\varepsilon^{-2}\log^2 n)$. Our algorithm avoids the use of $O(\varepsilon^{-1}\log n)$ copies of the graph, resulting in a smaller space consumption and an algorithm with provably low recourse. 
 
\end{abstract}
\newpage

\section{Introduction \& Related work}

In the edge-colouring problem one has to assign colours to the edges of a simple graph such that no edges, sharing an endpoint, receive the same colour. More formally speaking, given a graph $G = (V,E)$ on $n$ vertices and $m$ edges a (proper) $k$-edge-colouring of $G$ is a function $c:E \mapsto [k]$ satisfying that for any two edges $e,e'$ such that $e \cap e' \neq \emptyset$, we have $c(e) \neq c(e')$. 
Vizing famously showed that if $\Delta$ is the maximum degree of $G$, then $G$ has a $(\Delta+1)$-edge-colouring~\cite{vizing1964estimate}, and his proofs extends to give a polynomial time algorithm for computing such a colouring. There exist graphs for which this is tight, but there also exist graphs where $\Delta$ colours suffice. Note that since any such graph has a vertex of degree $\Delta$, one clearly needs at least $\Delta$ colours to properly colour its edges.
It was shown by Holyer that it is NP-complete to distinguish between when $G$ can be $\Delta$-edge coloured and when it cannot~\cite{Ian}.

Vizing actually showed something slightly stronger than the existence of a $(\Delta+1)$-edge-colouring: he showed that given a proper $(\Delta+1)$-edge-colouring of a subgraph $G'$ of $G$, we may extend it to a proper $(\Delta+1)$-edge-colouring of a larger subgraph $G' \subset G''$ of $G$ by recolouring at most $O(\Delta + n)$ edges of $G$. These edge form what we will refer to as an \emph{augmenting subgraph} that is a subgraph such that only recolouring edges inside this subgraph allows us to extend the colouring of an edge. 
Small augmenting subgraphs play an important role if one wants to extend the edge-colouring by colouring multiple uncoloured edges in parallel.

Recently there has been some exciting work proving better upper bounds for the sizes of such subgraphs. As one can perhaps imagine, the existence of small augmenting subgraphs is useful for constructing algorithms, and so these results arise in different models of computation.
In order to design an efficient algorithm for colouring dynamic graphs, Duan, He and Zhang~\cite{duan} essentially showed that if one allows $(1+\varepsilon)\Delta$ colours, then there must exist augmenting subgraphs for every uncoloured edge of size $O(\poly(\Delta, \log n))$. 
With the goal of designing an efficient LOCAL algorithm, Bernshteyn~\cite{BERNSHTEYN}, inspired by the approach of Greb\'\i k and Pikhurko~\cite{grebik2020measurable}, showed that in fact there exists augmenting subgraphs of size $O(\poly(\Delta)\log^2 n)$ even if one only has $\Delta+1$ colours available. 
All these result rely on the idea that if the augmenting subgraph constructed in Vizing's Theorem -- let us call such a subgraph a Vizing chain -- is too big, then one can truncate it early and create a new one. 
This idea can be repeated until a short Vizing chain is found.
If the Vizing chains do not overlap, the above papers show via probabilistic arguments that there has to exist a choice of chains and truncation points that terminates in a small augmenting subgraph quickly. Bernshteyn coined his construction by the name \emph{multi-step Vizing chain}. 

From the lower bound side, it was shown by Chang, He, Li, Pettie and Uitto~\cite{pettie} that augmenting subgraphs sometimes need to have size at least $\Omega(\Delta \log \frac{n}{\Delta})$. This lower bound immediately carries over to the round complexity of any LOCAL algorithm that works by naively extending colours through augmenting subgraphs. 

In this paper, we consider the most general definition of a multi-step Vizing chain. In particular, we will allow subsequent Vizing chains to overlap in edges. 
We will then consider constructions of various subclasses of Vizing chains. 

The different constructions from above all give rise to a subclass of Vizing chains which we will call \emph{non-overlapping}. 
We will also consider non-overlapping multi-step Vizing chains, but in contrast to Duan, He and Zhang~\cite{duan} and Bernshteyn~\cite{BERNSHTEYN} we view them through a non-probabilistic lens and show that they can be used to construct augmenting subgraphs of size $O(\poly(\Delta)\log n)$. This shows that edge-colourings of graphs with low degree are locally extendable -- in the sense that they can be extended through augmenting subgraphs in the close neighbourhood of the uncoloured edges. The existence of such subgraphs also allows us to give an improved LOCAL algorithm for computing $(\Delta+1)$ edge-colourings that falls into a long line of work where round complexity is parameterised by $\Delta$ and $\log n$ (see for instance~\cite{BERNSHTEYN,pettie,ghaffari2018deterministic,rozhovn2020polylogarithmic,su2019towards} for a subset) which we elaborate on below.

Moreover, we define another kind of multi-step Vizing chains that we call \emph{strictly local} Vizing chains and use them to prove the existence of an edge-colouring that takes the local structure around the edges into account. 
Namely, we will show the following conjecture of Bonamy, Delcourt, Lang and Postle~\cite{bonamy2020edge}.
\begin{theorem}[Local Vizing Theorem, Conjectured in~\cite{bonamy2020edge}]
For any graph $G$, there exists a proper edge-colouring of the edges of $G$ such that all edges $uv$ receives a colour in $\{1, \dots, 1+\max \{d(u), d(v)\}\}$.
\end{theorem}
Here each edge receives a list of colour, depending on the local structure around it, and one has to find a proper edge-colouring of the graph so that every edge receives a colour from its prescribed list. 
This sort of local generalisation is well-studied in the context of \emph{list colouring} where each element that is to be coloured, for instance the edges or the vertices of a graph, has to receive a colour from the prescribed lists. 
In the context of vertex colouring, Bonamy, Kelly, Nelson, and Postle recently studied list colourings in terms of local clique sizes~\cite{Bonamy2}, and Davies, de Joannis de Verclos, Kang, and Pirot studied the problem for triangle-free graphs~\cite{Pirot}. 

In the context of edge-colourings, local generalisations of list edge-colourings was studied by Borodin, Kostochka, and Woodall~\cite{Borodin} who showed a local generalisation of Galvin's theorem~\cite{DBLP:journals/jct/Galvin95} and by Bonamy, Delcourt, Lang, and Postle~\cite{bonamy2020edge} who, under certain degree conditions, showed a local version of Kahn's theorem~\cite{DBLP:journals/jct/Kahn96}. Work on local generalisations of edge-colourings also appeared in the work of Erd\H{o}s, Rubin, and Taylor~\cite{erdos1979choosability}. 

As hinted to earlier, the constructions and ideas from above are useful when designing efficient edge-colouring algorithms in various models of computation. We will consider two such models. 

\paragraph{LOCAL algorithms} 
In~\cite{linial1992locality} Linial introduced the LOCAL model of computation. Here one is given an input graph that is viewed as a communication network. Computation is performed at each vertex in synchronous rounds. During a round, vertices are allowed to exchange messages of unbounded size with their neighbours in the communication network as well as perform unbounded computation. 

edge-colouring has been widely studied in this model of computation. Alon, Babai and Itai~\cite{alon1986fast} and independently Luby~\cite{luby1985simple} gave a randomised $O(\log n)$ round algorithm for $\Delta + 1$ vertex colouring, which can be transformed into a $2\Delta-1$ edge-colouring algorithm by colouring the line graph. 
Goldberg, Plotkin and Shannon~\cite{goldberg1987parallel} improved the round complexity for graphs of small maximum degree by giving a randomised algorithm running in $O(\Delta^2 + \log^* n)$ rounds. 
After a long line of work Fischer, Ghaffari and Kuhn~\cite{fischer2017deterministic} provided a determinstic algorithm for $2\Delta-1$ edge-colouring a graph in $O(\poly \log n)$ rounds. 
Recently, this was improved by Balliu, Brandt, Kuhn and Olivetti~\cite{Brandt} who gave a deterministic algorithm for $2\Delta-1$ edge-colouring a graph in $\poly(\log \Delta)+O(\log^* n)$ rounds.

Much work has also been dedicated to going below $2\Delta-1$ colours (see for instance~\cite{pettie,ghaffari2018deterministic} for more extensive surveys).
Chang, He, Li, Pettie and Uitto~\cite{pettie} gave a randomized algorithm using $\Delta+ O(\sqrt{\Delta})$ colours running in $O(\poly(\Delta, \log \log n))$ rounds. 
Recently, Davies~\cite{Davies} gave a $\log^{O(1)} \log n$ round randomized algorithm for computing a $\Delta + o(\Delta)$ colouring.
Su and Vu~\cite{su2019towards} designed a randomised algorithm using only $\Delta + 2$ colours that runs in $O(\poly(\Delta, \log n))$ rounds.
Su \& Vu arrive at their $(\Delta + 2)$-edge-colouring algorithm by truncating Vizing chains at a randomly chosen spot, but instead of building a new one like Bernshteyn does~\cite{BERNSHTEYN}, they instead use a special colour to colour the now uncoloured edge. They show that with high-probability no vertex ends up with two neighbouring edges receiving this special colour. 
Due to a general derandomisation technique developed by Ghaffari and Rozho{\v{n}}~\cite{rozhovn2020polylogarithmic} this can be turned into a deterministic algorithm.
Other deterministic results include a $\lfloor \frac{3\Delta}{2} \rfloor$ edge-colouring algorithm due to Ghaffari, Kuhn, Maus and Uitto~\cite{ghaffari2018deterministic} as well as the already mentioned result of Bernshteyn~\cite{BERNSHTEYN}, who gave an algorithm using only $\Delta + 1$ colours. Both of these algorithms run in $O(\poly(\Delta, \log n))$ rounds.

\paragraph{Dynamic graph algorithms}
In the dynamic graph setting, one attempts to maintain solutions to graph problems as the graph is subjected to insertions and deletions of edges. 
The goal is that the time spent updating the solution should be significantly faster than naively re-running a static algorithm from scratch each time an edge is inserted or deleted.
For many problems one aims at achieving $O(\poly \log n)$ update time. 

From a static point-of-view Gabow, Nishizeki, Kariv, Leven, and Tereda~\cite{Gabow} showed how to compute a $(\Delta + 1)$-edge-colouring in $\tilde{O}(m \sqrt{n})$ or $\tilde{O}(m \Delta)$ time. This improved over the $O(mn)$ time algorithm that follows directly from the work of Vizing~\cite{vizing1964estimate}. 
Alon~\cite{alon2003simple} and Cole, Ost and Schirra~\cite{cole2001edge} gave a near linear-time algorithm for $\Delta$-edge-colouring bipartite graphs. Via a reduction due to Karloff and Shmoys~\cite{karloff1987efficient}, these algorithms yield a $3 \lceil \frac{\Delta}{2} \rceil$ edge-colouring algorithm for general graphs. 
When $\Delta = \Omega{(\varepsilon^{-1} \log n)}$, Duan, He and Zhang~\cite{duan} give a $O(\varepsilon^{-2}m \log^6 n)$ time algorithm for computing a$(1+\varepsilon)\Delta$-edge-colouring.

The dynamic case has recently received attention and seen some development. 
Barenboim and Maimon~\cite{barenboim2017fully} gave a dynamic algorithm with $O(\sqrt{\Delta})$ worst-case update time using $O(\Delta)$ colours. This was subsequently improved by Bhattacarya, Chakrabarty, Henzinger and Nanongkai who gave a fully dynamic algorithm for computing a $2\Delta-1$ edge-colouring with worst case $O(\log \Delta)$ update time. 
Finally, Duan, He and Zhang~\cite{duan} reduced the number of colours even further. They designed an algorithm using $(1+\varepsilon)\Delta$ colours when $\Delta = \Omega(\varepsilon^{-2}\log^2 n)$ achieving an amortised update time of $O(\varepsilon^{-4}\log^8 n)$. 
This approach changes between $O(\varepsilon^{-1} \log n)$ copies of the graph, resulting in a $O(\varepsilon^{-1} \log n)$ space overhead and a recourse which is not bounded. 
Bhattacarya, Grandoni, and Wajc~\cite{doi:10.1137/1.9781611976465.168} studied the problem of dynamically maintaining edge-colourings with low recourse. 
More specificically, they show how to maintain a $(1+\varepsilon)\Delta$-edge-colouring with high probability with only $O(\poly{\varepsilon^{-1}})$ expected recourse in dynamic graphs of maximum degree $\Delta = \Omega(\poly(\varepsilon^{-1}, \log n))$. 

\subsection{Results}
We show the following local version of Vizing's Theorem originally conjectured by Bonamy, Delcourt, Lang and Postle~\cite{bonamy2020edge}:
\begin{theorem}[Local Vizing Theorem] \label{thm:LVT}
For any graph there exists a proper edge-colouring such that every edge $uv$ receives a colour from the list $L(uv) = \{1, \dots, 1+\max \{d(u), d(v)\}\}$.
\end{theorem}
This Theorem immediately implies Vizing's Theorem since $\Delta = \max \limits_{w} d(w)$. 
Then we shift our attention from local lists to extending colourings by recolouring edges locally. We show that the lower bound due to Chang, He, Li, Pettie and Uitto~\cite{pettie} is tight up to constant and $\Delta$ factors:
\begin{theorem} \label{thm:LOCALVT}
For any graph $G$ endowed with a proper partial $(\Delta+1)$-edge-colouring $c$, and for any edge left uncoloured by $c$, there exists an augmenting subgraph containing at most $O(\Delta^7 \log n)$ edges.
\end{theorem}
Finally, we extend the techniques used to derive the two combinatorial results from above to devise efficient algorithms. First of all, the following corollary follows from our proof of Theorem~\ref{thm:LVT}:
\begin{corollary}\label{cor:algoLVT}
There is an algorithm that runs in $O(n^2 \Delta)$-time for computing an edge-colouring such that every edge $uv$ receives a colour from the list $L(uv) = \{1, \dots, 1+\max \{d(u), d(v)\}\}$. 
\end{corollary}
We improve the round complexity of the best LOCAL algorithm for computing a $(\Delta + 1)$-edge-colouring from\footnote{The exponent is not analysed in the paper~\cite{BERNSHTEYN}. From our understanding, the best we could extract is $\log^{11}(n)$.} $\tilde{O}(\poly(\Delta,\log n))$ to $\tilde{O}(\poly(\Delta)\log^6 n)$.
\begin{theorem} \label{thm:algoLOCALVT}
There is a deterministic LOCAL algorithm that computes a $\Delta+1$ edge-colouring in $\tilde{O}(\poly(\Delta)\log^6 n)$ rounds\footnote{Here $\tilde{O}(x)$ suppress $O(\poly \log x)$ factors.}.
\end{theorem}
Finally, we give a dynamic algorithm which adapts to the current value of $\Delta$ without the need to change between multiple copies of the graph.
\begin{theorem}\label{thm:dynVT}
Let $G$ be a dynamic graph subject to insertions and deletions with current maximum degree $\Delta$. There exists a fully-dynamic and $\Delta$-adaptive algorithm maintaining a proper $(1+\varepsilon)\Delta$-edge-colouring with $O(\varepsilon^{-6} \log ^6 \Delta \log^9 n)$ worst-case update time with high probability.
\end{theorem}

\subsection{High-level overview}
In this section, we give an informal description of our techniques and approaches, before we give the precise definitions and technical proofs later. 
As mentioned in the introduction, the starting point is the idea of 'gluing' together several Vizing chains to form subgraphs with interesting properties. In the literature there are several ways of constructing Vizing chains, but we will be working with the following:
a Vizing chain  consists of two parts: a) a \emph{fan} $F$ of edges and b) a bichromatic path $P$ consisting only of edges coloured with 2 colours, let us call them $\kappa_1$ and $\kappa_2$. 
The fan $F$ consists of a center vertex $u$ together with edges $uw_1, \dots, uw_k$ such that the colour of $uw_{i+1}$ is available at $w_i$, meaning that no edges incident to $w_i$ has the colour $c(uw_{i+1})$.
The bichromatic path $P$ has to then begin at $w_k$ and consist only of edges coloured $\kappa_1$ and $\kappa_2$, where $\kappa_1$ is available at $u$ and $\kappa_2$ is available at $w_k$. 
Given a Vizing chain, we can \emph{shift} it by recolouring edges in $F$ and subsequently in $P$ as shown in Figure~\ref{fig:shiftVC} to the left. This will transform one proper partial edge-colouring (the \emph{pre-shift} colouring) into another proper partial edge-colouring (the \emph{post-shift} colouring).
\begin{figure}%
    \centering
    \subfloat{{\includegraphics[width=4cm]{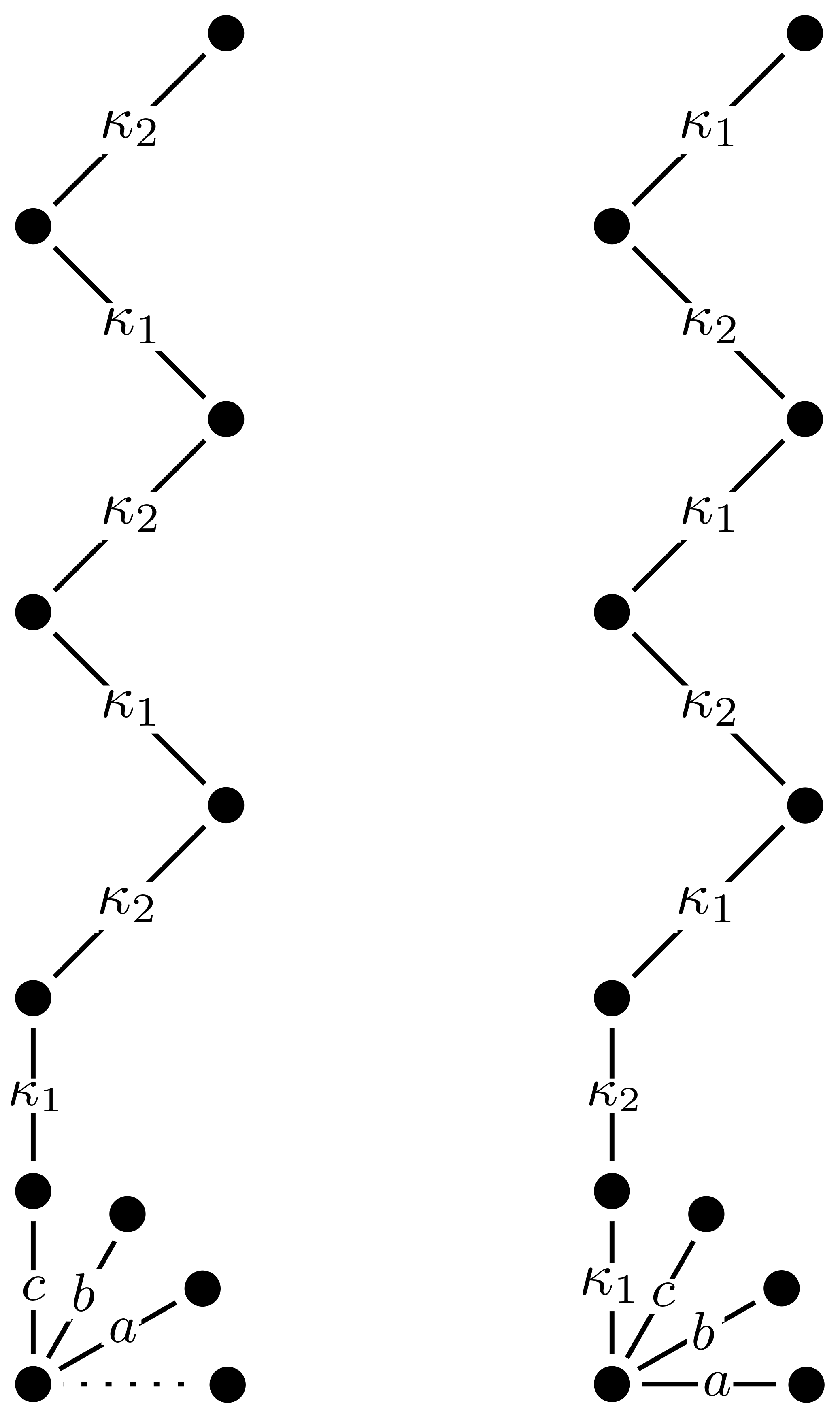} }}%
    \hspace{30mm}%
    \subfloat{{\includegraphics[width=4cm]{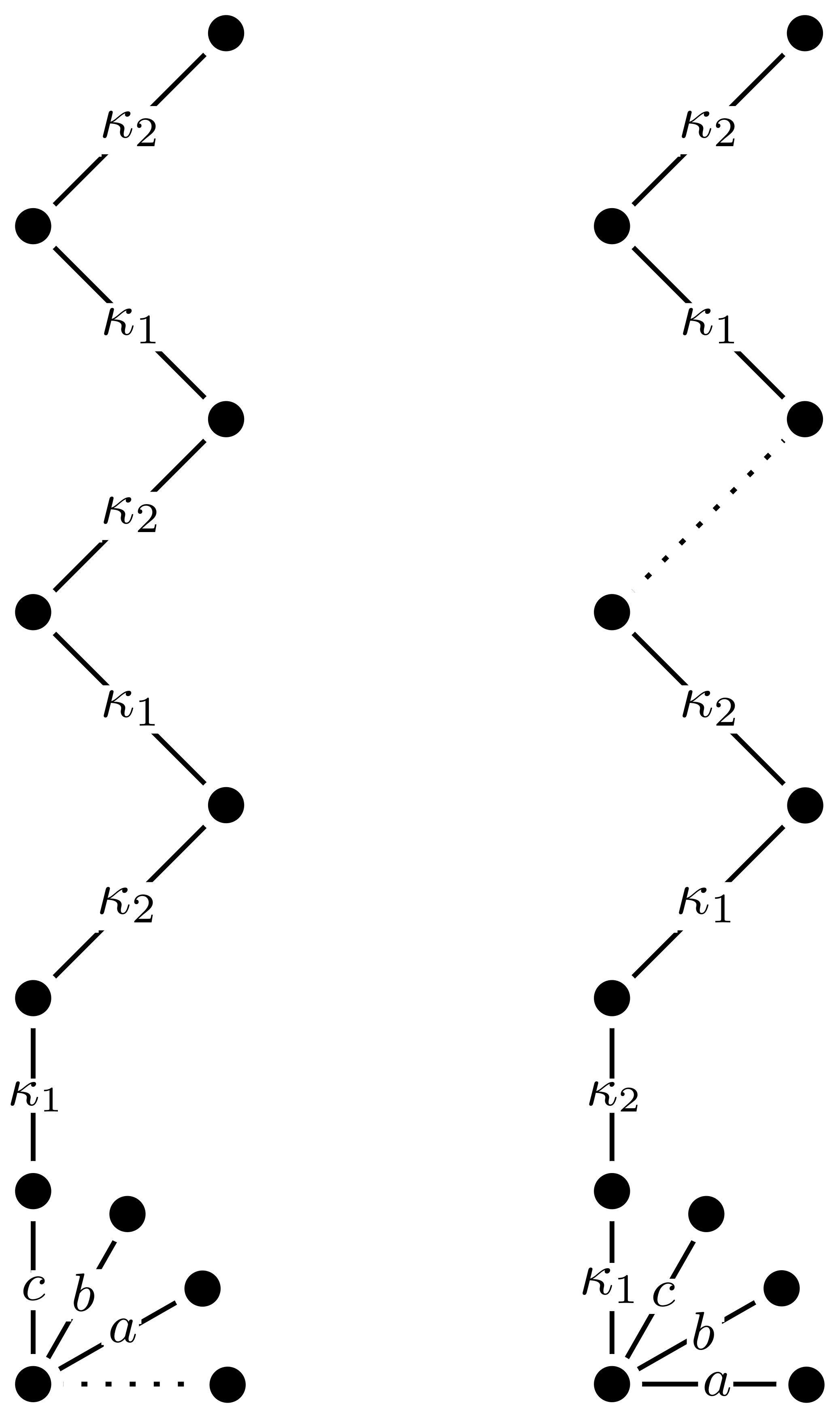} }}%
    \caption{An augmenting Vizing chain is shifted on the left, and a truncated Vizing chain is shifted on the right.}
    \label{fig:shiftVC}%
\end{figure}
In this example, we see that the shift allows us to extend the colouring to colour one extra edge. The idea is that if $P$ is not too long, we can shift the Vizing chain to extend our colouring, and if $P$ is too long, we can \emph{truncate} it by uncolouring an edge along $P$ and only shift the first part of $P$ as shown in Figure~\ref{fig:shiftVC} to the right. 

Then we may build a new Vizing chain on top of the new uncoloured edge in the hope that this chain has a shorter length. We may iterate this to construct \emph{multi-step Vizing chains} (see Figure~\ref{fig:multistep}). This approach is originally due to Bernshteyn~\cite{BERNSHTEYN}, but where Bernshteyn only uses Vizing chains that cannot overlap in edges, we will relax this requirement and allow the concatenation of all types of Vizing chains. 
\begin{figure}%
    \centering
    \subfloat{{\includegraphics[width=14cm]{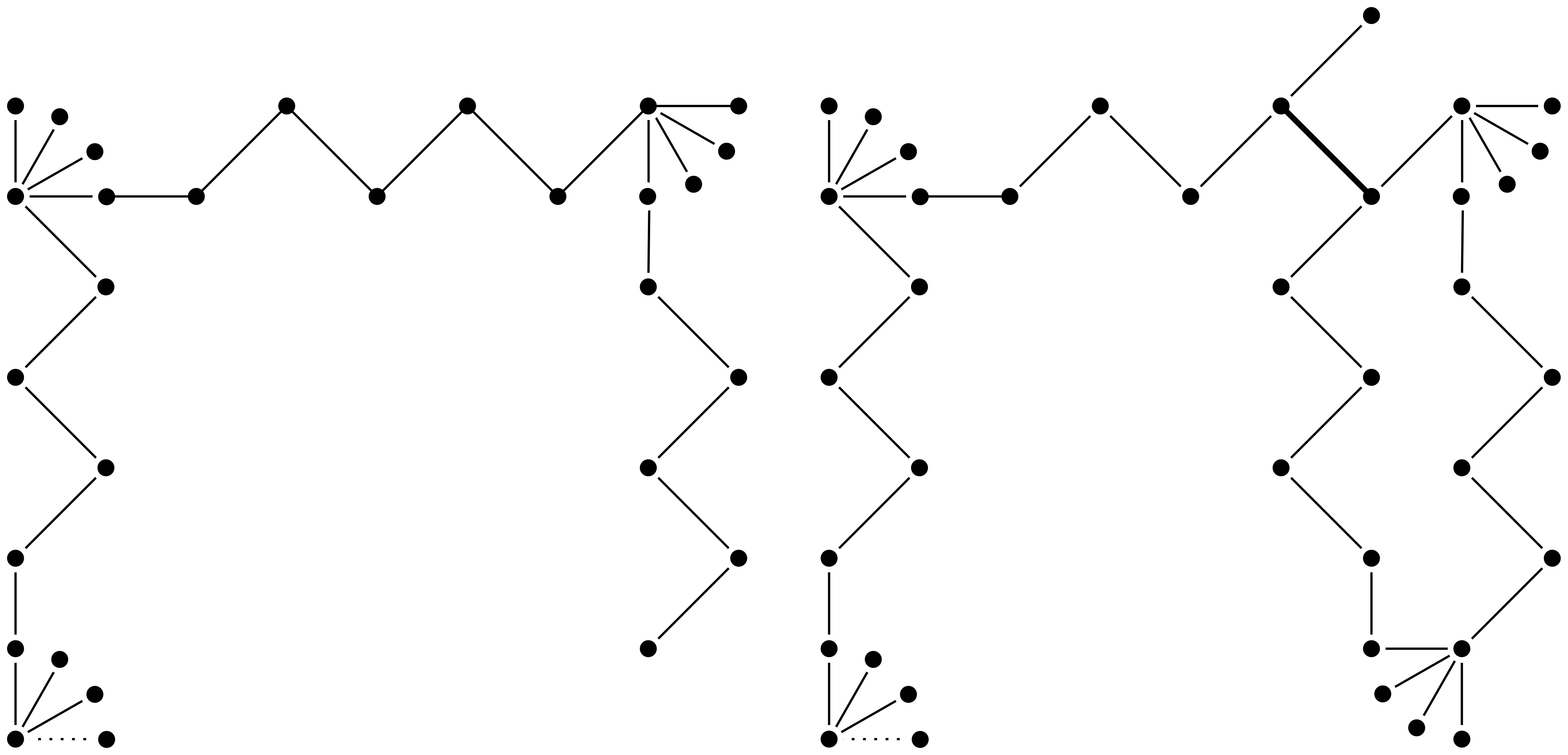} }}%
    \caption{On the left, we see a non-overlapping multi-step Vizing chain. On the right, we see a multi-step Vizing chain overlapping in a single edge.}
    \label{fig:multistep}%
\end{figure}
First we turn our attention to what we will call \emph{strictly local} Vizing chains. In these Vizing chains, we require that the post-shift colouring is strictly local, meaning that it satisfies the condition of the Local Vizing Theorem for every coloured edge. This means that a coloured edge $uv$ has to receive a colour from $[\max\{d(u),d(v)\}+1] = \{1, 2, \dots, \max\{d(u),d(v)\}+1\}$. 
The idea is that if we can show how to construct augmentable strictly local Vizing chains, then we can use them to show Theorem~\ref{thm:LVT} as follows: we begin with an uncoloured graph. This colouring is trivially strictly local since it does not assign a colour to any edges. By continually constructing and shifting augmenting strictly local Vizing chains, we iteratively build a proper non-partial strictly local colouring.

Hence, we shift our attention to constructing such an augmenting subgraph. Given a strictly local pre-colouring, we construct a strictly local Vizing chain by, first of all, demanding that every vertex $v$ picks its available colour in $[1+ d(v)]$, when we are constructing the fan and the bichromatic path. 
This means that we enforce that $c(uw_{i+1}) \in [d(w_i)+1]$ and that $\kappa_1 \in [d(u)+1]$ and $\kappa_2 \in [d(w_{k})+1]$, when we construct our fan $uw_1, \dots, uw_k$ and our $(\kappa_1,\kappa_2)$-bichromatic path at $w_k$.
Note that this is always possible, as any vertex $v$ is incident to at most $d(v)$ colours. 
Second of all, we will truncate the Vizing chain at the first edge we meet that would violate strict locality were we to shift the Vizing chain. 
The key insight is that shifting such a Vizing chain lowers the total number of times some colour $\kappa$ is available at a vertex $v$ where $\kappa \leq 1+d(v)$.
This allows us to define a potential that drops at the shift of such a strictly local Vizing chain. 
Hence, we may glue together enough of these Vizing chains until we reach a point where the path contains no edges that would invalidate our colouring after being shifted. This will then form our desired augmenting subgraph. 

After this, we turn our attention to what we call \emph{non-overlapping} Vizing chains. These are multi-step Vizing chains, where the $1$-step Vizing chains that make up the multi-step Vizing chains only are allowed to overlap in the edges, where they are glued together. 
This set of chains contains the Vizing chains constructed in~\cite{BERNSHTEYN, duan, grebik2020measurable}.
The key idea behind enforcing that the Vizing chains are non-overlapping is that it allows one to reason about how many Vizing chains can reach a certain point or a certain edge in the graph. 
The main observation is that if we consider a single point $v$, then $v$ can be reached by at most $O(\Delta^2)$ bichromatic paths, and hence by at most $O(\Delta^3)$ 1-step Vizing chains. Indeed, a point lies in at most $O(\Delta^2)$ different bichromatic paths, since we have $\Delta+1$ choices for each of the two colours classes that make up the path. Since each bichromatic path can only be part of 1-step Vizing chains neighbouring its endpoints, we have at most $O(\Delta^3)$ such 1-step Vizing chains.
This idea can be applied iteratively to show that at most $O(\Delta^3)^{i}$ $i$-step non-overlapping Vizing chains can reach the point $v$. 
Bernshteyn~\cite{BERNSHTEYN}, who was inspired by Greb\'\i k and Pikhurko~\cite{grebik2020measurable}, used this to control the error rate of a probabilistic construction. 
Bernshteyn argued that if one constructs a multi-step Vizing chain by picking the truncation edges uniformly at random, then the construction has to terminate fast. 
Indeed, one is increasingly unlikely to continually pick right extension points leading one to a specific vertex $w$ at the $T$'th step. 
In particular, after enough steps, the probability of reaching any point is vanishing.
This shows the existence of a relatively short augmenting Vizing chain. 
This leaves only the question of how to ensure that the Vizing chains are non-overlapping. 
Here it is difficult to control the probability of a certain truncation edge leading to an overlapping Vizing chain. 
To overcome this challenge, Bernshteyn increases the step length to also depend on the number of steps taken. Since one has to take $\log n$ steps, it is inevitable that this approach leads to Vizing chain of length $\Omega(\log^2 n)$

We proceed completely differently compared to Bernshteyn. Instead of arguing probabilistically, we will take a density-based approach. The idea is to first construct a Vizing chain as is usually done when proving Vizing's theorem.
We then let $R_1$ be the set containing the first $\ell$ vertices of the bichromatic path of the Vizing chain. We say that these vertices are \emph{reachable} via a $1$-step Vizing chain, since we can reach them by truncating the Vizing chain at an edge incident to them, and shifting the uncoloured edge to be incident to such a vertex. 
If the bichromatic path of this chain is too long, say it has length $> \ell$, we consider the act of truncating this Vizing chain at any of the first $\ell$ edges and \emph{extending} it to a $2$-step Vizing chain through one of the \emph{points} in $R_1$ using the construction of Greb\'\i k and Pikhurko~\cite{grebik2020measurable}.
This allows us to get a second Vizing chain build on the truncation edge such that 1) the fan from this Vizing chain does not overlap with the bichromatic path of the first Vizing chain, and 2) the new Vizing path $P_2$ consists of the same colours as $P_1$ or of two completely different colours. 
We then examine all of these bichromatic paths. 
Some of them may overlap the first Vizing chain, and this renders the respective Vizing chains, and the points that we build them on, unfit to be used as an extension point. We let $N_1$ contain all of the points in $R_1$ that are not unfit for extension. Hence, we can safely extend our multi-step Vizing chain through these extension points.

Therefore, we can consider the set of points reachable via $1$-step Vizing chains of the points in $N_1$. These are exactly the points that we can reach via a $2$-step non-overlapping Vizing chain from the first point. 
If all of these Vizing chains have length $>\ell$, we have not yet found a short augmenting non-overlapping Vizing chain. Therefore, we can let $R_2$ be the set of points that we can reach through non-overlapping $2$-step Vizing chains, and similarly $N_2 \subset R_2$ be the set of points which we can safely extend our Vizing chain through. Similarly, if we fail to find a short, augmenting $i$-step Vizing chain, we can recursively define $R_i$ to be the set of points we can reach through $i$-step non-overlapping Vizing chains, and $N_i \subset R_i$ the set of points through which we can safely extend our Vizing chains. 

The idea is then to show that both $R_i$ and $N_i$ will grow exponentially in $\Omega(\frac{\ell}{\Delta^3})$ for $\ell \in \poly(\Delta)$ (see Figure~\ref{fig:neighbourhoods} for an illustration). 
% Include calculations.
\begin{figure}%
    \centering
    \subfloat{{\includegraphics[width=6cm]{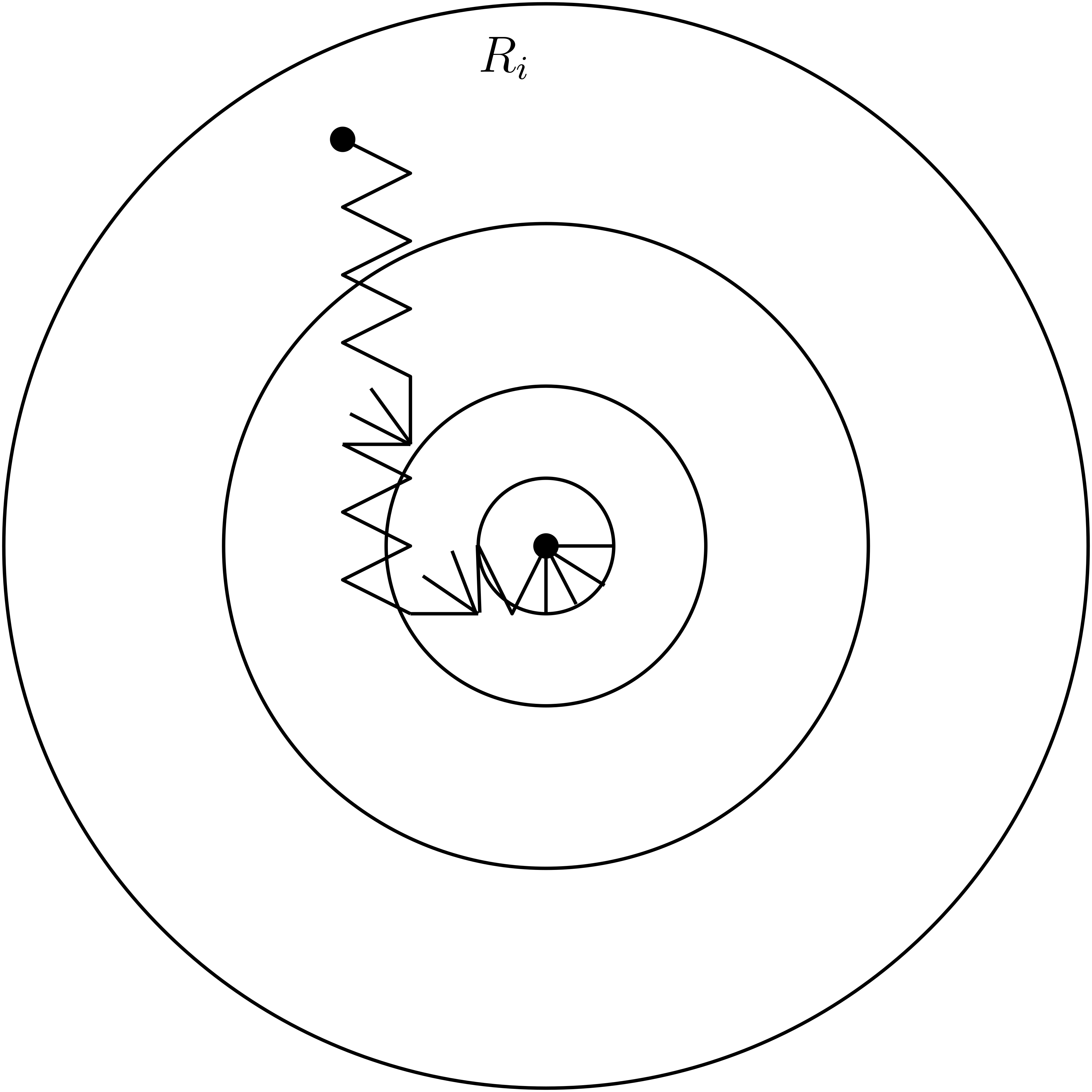} }}%
    \caption{The number of reachable vertices through non-overlapping Vizing chains of length $i$ form $R_i$. We will later show that $R_i$ grows exponentially as a function of $i$.}
    \label{fig:neighbourhoods}%
\end{figure} 
This implies that $i \leq \log n$, and so we can conclude that there has to exist some augmenting Vizing chain that takes few steps and has a very short step-length. The idea is to show that if there is no short augmenting and non-overlapping $j$-step Vizing chain for any $j\leq i$, then $|R_j| = \Omega(\frac{\ell}{\Delta^3} |N_{j-1}|)$ and $|N_{j}| \geq \frac{|R_{j}|}{2}$ for all $j \leq i$. Combining these things inductively, then shows the exponential growth. 

To this end, we consider the act of extending exactly one $(i-1)$-step Vizing chain through a point in $N_{i-1}$. We may have multiple options for which $(i-1)$-step Vizing chain to extend, but we just pick one arbitrarily. 
Consider the subgraph $H$ spanned by the edges of the bichromatic paths of these Vizing chains. The key observations are that $V(H) \subset R_{i} \cup N_{i-1}$ and that $\ell \cdot{} |N_{i-1}| = O(\Delta^2|E(H)|)$. Indeed, each vertex in $N_{i-1}$ contributes $\ell$ edges to $H$, and in this way we only count each edge of $H$ $O(\Delta^2)$ times. This is because a coloured edge is part of at most $O(\Delta)$ bichromatic paths -- one for each other choice of colour -- and each bichromatic path can be part of at most $O(\Delta)$ Vizing chains, namely Vizing chains situated at neighbours of the endpoints of the bichromatic paths. Since the density of $H$ satisfies $\rho(H) = \frac{|E(H)|}{|V(H)|} \leq \frac{\Delta}{2}$, we find that 
\[
|R_i| = \Omega(\frac{\ell}{\Delta^3} |N_{i-1}|)
\]
This leaves only the task of counting how many points in $R_i$ that are unfit for extension. 
To do this, we recall that at most $O(\Delta^3)$ Vizing chains can go through a point, and so at most 
\[
O\paren{\Delta^4 \cdot{}\paren{\sum \limits_{k=1}^{i-1} |R_{k}| + |N_{i-1}|}}
\]
points in $R_{i}$ can extend to overlapping Vizing chains. This follows from the fact that for a Vizing chain to be overlapping, it certainly has to go through either its own fan situated at a point in $N_{i-1}$ or through a point in the one-hop neighbourhood of $\sum \limits_{k=1}^{i-1} |R_{k}|$. 
By choosing $\ell$ properly, we find that $|N_{i}| \geq \frac{|R_{i}|}{2}$.

\paragraph{Distributed algorithm} In order to turn this into a distributed algorithm, we first note that if augmenting subgraphs are vertex disjoint, then we may augment them in parallel, since changing the colours within one subgraph does not change the colour of any edges incident to vertices in the other subgraphs. 

The most difficult part left is to show that there actually exist many such small vertex disjoint subgraphs that we can augment in parallel.
This is because of the following reduction due to Bernshteyn~\cite{BERNSHTEYN}. One can define a hypergraph on $V(G)$ by including a hyperedge for each small augmenting subgraph containing exactly the vertices of the small augmenting subgraph. Given that we do not consider too many augmenting subgraphs and that there exists a large vertex disjoint subset of the subgraphs, one can find such a large vertex disjoint subset by computing an (approximation) of the maximum matching in the hypergraph using for instance the algortihm due to Harris~\cite{harris2019distributed}. 

Hence, we turn our attention to showing that if we include all augmenting Vizing chain with both a small number of steps and a small step-length, then we will have a large set of vertex disjoint augmenting subgraphs. 
The main intuition that we want to formalise is that even if we allow some adversary to remove $O(\Delta^4)^{i+1}$ from $N_i$ at each step above, the construction from above still goes through. 
This motivates the definition of what we call \emph{family-avoiding} avoiding Vizing chains. 
Given a family of sets $\mathcal{F} = \{\mathbb{F}_{j}\}_{j = 1}^{\infty}$, we say that a multi-step Vizing chain avoids $\mathcal{F}$ if the $i^{\textrm{th}}$ Vizing chain is not extended through a point belonging to $\{\mathbb{F}_j\}_{j = 1}^{i-1}$. 
We also say that $\mathcal{F}$ is $k$-bounded if for all $j$ we have $|\mathbb{F}_j| \leq k^j$. 
By above, our previous construction can actually be made to avoid any $O(\Delta^4)$-bounded family. 

This concept allows us to build a large set of vertex disjoint small augmenting subgraphs.
First we pick an arbitrary such vertex disjoint subgraph. Having picked the first $j$ vertex disjoint subgraphs, we pick the $(j+1)^{\textrm{th}}$  as follows: for any uncoloured edge $e$, we build a family to avoid, when we use the previous construction to build small augmenting subgraphs. 
In particular, we add a vertex to $\mathbb{F}_{j-1}(e)$, if a non-overlapping $j$-step Vizing chain beginning at $e$ can reach any of the vertex disjoint subgraphs that we already picked. 
If this is the case, we add the vertex that the $j^{\textrm{th}}$  Vizing chain was extended through to $\mathbb{F}_{j-1}(e)$. 
Next we note that if the total size of our subgraphs is $T$, then, as we saw earlier, we put at most $T\cdot{}O(\Delta^4)^{i+1}$ points into some $\mathbb{F}_i$ for all $i$. Indeed, at most $T\cdot{}O(\Delta^4)^{i}$ many $i$-step Vizing chains can reach $T$. 
In the worst case, the points are added in a way that maximizes the number of families $\mathcal{F}(e)$ that become \emph{unavoidable}, meaning that they are not sufficiently bounded for our construction to be able to avoid them. 
However, since the number of points we can avoid grows faster than the number of points we have to avoid, the number of edges for which $\mathbb{F}_{j}(e)$ is too large falls exponentially in $j$, and so in total at most $O(\Delta^4 \cdot{} \Delta \cdot{} T)$ edges have unavoidable families. 
We pick the $(j+1)^{\textrm{th}}$  subgraph as a small augmenting subgraph of an edge that avoids the family constructed for it above. This means that this subgraph has to be vertex disjoint form the subgraphs we already picked, since we are certain to never extend through a point that can reach any of these subgraphs in $1$-step.

\paragraph{Dynamic algorithm} 
%include calculations.
Let us first recall the approach of Duan, He and Zhang~\cite{duan}. It has 3 steps: first they make an amortised reduction to the easier non-adaptive version of the problem, where we may assume $\Delta$ is fixed. 
They do this by using multiple copies of the graph -- each with a different upper bound on the maximum degree. The second step is to reduce the problem to one where $\Delta = O(\frac{\log n}{\varepsilon})$. 
This is done by sampling a palette $S \subset [(1+\varepsilon)\Delta]$ such that every vertex in $G$ has an available colour in $S$. Note that it is crucial to have $(1+\varepsilon)\Delta$ colours for this step be efficient. 
The third and final step is to show how to insert an uncoloured edge into such a low-degree graph. 
To do this, Duan, He and Zhang~\cite{duan} show that one can construct non-overlapping multi-step Vizing chains by using a new, disjoint palette for each step, and that short Vizing chains of this type have a good probability of being augmenting. 
Since they need to sample disjoint palettes for each step, they require that $\Delta = \Omega(\frac{\log^2 n}{\varepsilon^2})$. 

Our approach will circumvent the first step. Namely, we will maintain a colouring where we can control the number of edges coloured with a specific colour. In particular, for each colour $\kappa$ we will require that the number of edges coloured $\kappa$ is upper bounded by the number of vertices $v$ where $\kappa \in [(1+\varepsilon)d(v)]$. In fact, we will even maintain a slightly stricter invariant. 
To accomplish this we define a notion of $\varepsilon$-local Vizing chains, where we require all available colours of a vertex $v$ to be in $[(1+\varepsilon)d(v)]$. 
We show that shifting such Vizing chains does not violate the invariant from above, and so this allows us to reduce deletions to insertion via a very simple scheme that recolours $O(1)$ edges in order to accommodate a deletion. 

To realise this approach, we need to sample $\varepsilon$-local palettes. These are palettes $S$ where every vertex $v$ has an available colour in $S \cap [(1+\varepsilon)d(v)]$. 
We show that one can sample such a palette by combining several sub-palettes sampled from intervals of length powers of $2$. 
Finally, we use the insights from above to slightly alter an algorithm of Bernshteyn~\cite{BERNSHTEYN} so that it constructs augmenting, $\varepsilon$-local and non-overlapping Vizing chains. This allows us to get rid of the assumption that $\Delta = \Omega(\frac{\log^2n}{\varepsilon^2})$.

\subsection{Outline of the paper}
In Section~\ref{sec:prelim}, we recall notation and preliminaries. We then proceed to discuss different versions of multi-step Vizing chains in Section~\ref{sec:VC}
In Section~\ref{sec:app1}, we prove a local version of Vizing's Theorem i.e.\ Theorem~\ref{thm:LVT}. In Section~\ref{sec:app2}, we show our new upper bound on the size of augmenting subgraphs as described in Theorem~\ref{thm:LOCALVT}. We then, in Section~\ref{sec:app3}, use these insights to give a faster LOCAL algorithm for $(\Delta + 1)$-edge-colouring i.e.\ we show Theorem~\ref{thm:algoLOCALVT}. Finally, we describe our dynamic algorithm and show Theorem~\ref{thm:dynVT} in Section~\ref{sec:dynCol}. 

\section{Preliminaries \& Notation} \label{sec:prelim}
For an integer $t$, we let $[t] = \{1, 2, \dots, t \}$ denote the numbers in $\mathbb{N}$ that are greater than $0$ but at most $t$.
We let $G = (V,E)$ be a graph on $n$ vertices and $m$ edges. A subgraph $G' \subset G$ of $G$ is a graph such that $V(G') \subset V(G)$ and $E(G') \subset E(G)$. The \emph{$t$-hop neighbourhood} of $G'$, $N^{t}(G')$, are all of the vertices in $G$ of distance at most $t$ (in $G$) to a vertex in $G'$. 
A \emph{proper (partial) $k$-edge-colouring} $c$ of $G$ is a function $c:E(G) \mapsto [k] \cup \{\neg\}$ such that two coloured edges $e,e'$, sharing an endpoint, does not receive the same colour. Here an edge $e$ such that $c(e) = \neg$ is said to be uncoloured. Two uncoloured edges may share an endpoint in a proper partial colouring.
If no edges are uncoloured, we say that the colouring is a proper $k$-edge-colouring. Given a \emph{proper (partial) $k$-edge-colouring} $c$ and a set of colours $S \subset [k]$, we let $G[S]$ be the graph induced by the set of edges coloured with a colour from $S$. Unless otherwise stated, we assume from now on that any edge-colouring is a $(\Delta + 1)$ (partial) edge-colouring. 

A hypergraph $H$ is a graph where edges may be any set from the power set of $V(H)$, $E(H) \subset \mathcal{P}(V(H))$. The \emph{rank} of a hypergraph is the cardinality of the largest edge. The degree of a vertex in $H$ is the number of edges that contain $v$. A \emph{matching} in $H$ is as set of pairwise disjoint hyperedges. A maximum matching is a matching of maximum size. 

Given a proper (partial) edge-colouring of a graph $G$, we define the set $A(v)$ of \emph{available} colours at a vertex $v$ to consist of precisely the colours that no edge incident to $v$ has received.

\subsection{Chains and shifts}
In order to extend a partial edge-colouring, we might have to move between different edge-colourings obtained by re-colouring a subset of the edges. 
In order to describe this process formally, we briefly recall the chain and shift terminology as it was used by for example Bernshteyn~\cite{BERNSHTEYN} (see~\cite{BERNSHTEYN} for a more in-depth treatment). 

Let $e_1, e_2 \in E(G)$ be two adjacent edges in a graph $G$, and let $c$ be a proper partial colouring of $G$. Then we can define a colouring $\operatorname{Shift}(c,e_1,e_2)$ by setting: 
\[
\operatorname{Shift}(c,e_1,e_2)(e) = \begin{cases}
     c(e_2) &\quad\text{if } e = e_1 \\
     \neg &\quad\text{if } e = e_2 \\
    c(e) &\quad\text{if } e \notin \{e_1, e_2\} \\
     \end{cases}
\]
We say that such a pair of adjacent edges $e_1, e_2$ are $c$-\emph{shiftable} if $c(e_1) = \neg$ and $c(e_2) \neq \neg$ and the colouring $\operatorname{Shift}(c,e_1,e_2)$ defined above is a proper partial colouring. If the colouring $c$ is clear from the context, we will sometimes leave out this argument. 

A \emph{chain} $C$ of size $k$ is then a set of edges $C = (e_1, \dots, e_k)$ such that $e_{i}$ and $e_{i+1}$ are adjacent for all $i$ and $c(e_i) = \neg$ iff $i = 1$. 
For $0\leq j \leq k-1$, we can $j$-shift such a chain by performing $\operatorname{Shift}_{j}(c,C)$ defined as:
\begin{align*}
    \operatorname{Shift}_{0}(c,C) &= c \\
    \operatorname{Shift}_{i}(c,C) &= \operatorname{Shift}(\operatorname{Shift}_{i-1}(c,C),e_{i},e_{i+1})
\end{align*}
We say that $C$ is $c$-shiftable if every pair of edges $e_{i}, e_{i+1}$ is $\operatorname{Shift}_{i-1}$-shiftable. 
It is straightforward to check that $j$-shifting such a chain $C$ yields a proper colouring, where the unique uncoloured edge in $C$ is the edge $e_{j+1}$. 
We say that the chain ends at the edge $e_k$ and at the vertex $v$ that is shared by $e_k$ and $e_{k-1}$. 
Note that in the context of simple graphs, this vertex is well defined. 
We will refer to the act of $k$-shifting a chain of size $k$ as simply \emph{shifting} the chain. We let $c$ be the \emph{pre-shift} colouring of $C$ and $\operatorname{Shift}(c,C)$ the \emph{post-shift} colouring of $C$. 
 
The goal is to find a chain $C$ such that shifting $C$ yields an uncoloured edge $uv$ such that $A(u) \cap A(v) \neq \emptyset$ so that we may colour $uv$ to extend our partial colouring.
If this is the case, we say $C$ is an \emph{augmenting chain}, otherwise we refer to it is a \emph{truncated chain}.

The \emph{initial segment} of length $s$ of a chain $C$ is then the chain $C|s = (e_1, \dots, e_s)$. 
One can concatenate two chains $C_1 = (e_1, \dots, e_{s_1})$ and $C_2 = (e_{s_1},f_2, \dots, f_{s_2})$ to get the chain $C_1 + C_2 = (e_1, \dots, e_{s_1}, f_2, \dots, f_{s_2})$. 
Then in order to shift such a chain, we use $\operatorname{Shift}(c, C_1+C_2) = \operatorname{Shift}(\operatorname{Shift}(c, C_1), C_2)$.

Now we can define a \emph{fan} chain on $u$ to be a chain $F$ of the form $F = (uw_1, \dots, uw_k)$. Here we denote $u$ as the \emph{center} of the fan. 
For $k = 1$ this definition is ambiguous, but in any case where we do not specify the center, one can pick the center arbitrarily. 
Bernshteyn~\cite{BERNSHTEYN} showed that for such fans to be shiftable chains, we require $c(uw_{i+1}) \in A(w_i)$. 
We will let the colour of $c(uw_{i+1})$ be the \emph{representative available colour} at $w_i$. 
We will also pick a representative available colour at $w_k$ in $A(w_k)$. If this representative available colour is also either available at $u$ or if it is the representative available colour for some $w_j$ with $j < k$, then we say that the fan is \emph{maximal}. 

A $(\kappa_1, \kappa_2)$-bichromatic path is a subpath of $G$ consisting only of edges coloured $\kappa_1$ and $\kappa_2$. 
It is not difficult to check the following fact:
\begin{fact}
Let $P$ be a maximal $(\kappa_1, \kappa_2)$-bichromatic path in a proper colouring. Then the colouring obtained by colouring all of the $\kappa_1$-coloured edges with the colour $\kappa_2$, and all of the $\kappa_2$-coloured edges with the colour $\kappa_1$ is also a proper colouring.
We refer to this as \emph{shifting} the colours of $P$.
\end{fact}
A \emph{path chain} is a chain of the form $P = (e_1, \dots, e_k)$, such that the edges in the chain form a path in the graph. The \emph{length} of a path chain $P = (e_1, \dots, e_k)$ is $k-1$.
A \emph{bichromatic path chain} is then a chain $P = (e_1, \dots, e_k)$ that forms a path in $G$ such that $c(e_1) = \neg$ and the colour of $c(e_2) = c(e_4) = \dots = \kappa_1$ and $c(e_3) = c(e_5) = \dots = \kappa_2$, and furthermore such that $e_2$ is either the first or the last edge in some maximal $(\kappa_1, \kappa_2)$-bichromatic path. 
In order to highlight the colours present in the chain, we will sometimes refer to such a chain as a $(\kappa_1, \kappa_2)$-bichromatic path chain.

We will only consider $c$-shiftable bichromatic path chains, so we will always have that the bichromatic path chains are $c$-shiftable, even if we do not explicitly state it.
Also, we will sometimes say that a shiftable bichromatic path chain $P=(e_1, \dots, e_k)$ is \emph{rooted} at the endpoint of $e_2$ that is also the endpoint of the maximal bichromatic path containing $(e_2, \dots, e_k)$. 

Vizing originally showed how to choose the fan and path chains in order to construct what we will refer to as an \emph{augmenting subgraph} i.e.\ a subgraph such that by recolouring edges only in this subgraph, we can arrive at a new bigger proper partial edge-colouring. Vizing used the existence of such augmenting subgraphs to prove a theorem which we will recall below. It is illustrative to also give a proof of this theorem, so we will sketch a proof as well:
\begin{theorem}[Vizing's Theorem~\cite{vizing1964estimate}] \label{thm:VT}
Let $G$ be a graph of maximum degree $\Delta$. Then there exists a proper edge-colouring of $G$ using only colours in $[\Delta + 1]$.
\end{theorem}
\begin{proof}[Proof-sketch]
Assume we are given a proper partial edge-colouring of $G$ containing a maximum number of coloured edges. We will show that there exists an even bigger proper partial edge-colouring, and thus the Theorem follows by contradiction. Note that the same argument may be applied to extend any proper partial edge-colouring to one that is bigger.

If no edge is left uncoloured by the proper (partial) edge-colouring, there is nothing to prove, so we assume there exists at least one uncoloured edge $e = uv$. 
In order to extend the colouring to one where $e$ is also coloured, we construct a maximal fan chain centered at $u$ with the first edge being $e_1 = uv$. We set $v = w_1$. We pick an available representative colour in $A(v)$ for $v$. If this colour is available at $u$, we stop the construction. Otherwise, assume we have inductively constructed a fan chain $(uw_1, uw_2, \dots, uw_{j})$. 
we show how to pick $w_{j+1}$ or conclude that the fan is maximal. First, we check if $A(w_{j})$ contains a colour present at either $uw_{\ell}$ for some $\ell<j$ or at $u$. If this is the case, we may choose this colour as our representative available colour, and subject to the chosen representative colours our fan is now maximal.
If no such colour exists, we pick an arbitrary colour $\kappa$ in $A(w_{j})$, and let $w_{j+1}$ be the other endpoint of the $\kappa$-coloured edge incident to $u$. Hence $\kappa$ is then the representative available colour at $w_j$. Notice that this construction cannot continue indefinitely. 

In conclusion, the maximal fan consists of vertices $u,w_{1}, \dots, w_{k}$ and either the representative free colour at $w_k$ is also free at $u$, or it is identical to the representative free colour at $w_i$ for some $i < k-1$. Without loss of generality we may assume the latter case occurs. 
Indeed; suppose otherwise. Then we $k$-shift the fan and colour $uw_k$ with the now shared available colour of $w_k$ and $u$. 
So assume we are in the second case.

Now let $\kappa_1$ be a free colour at $u$ and $\kappa_2$ the representative free colour at $w_k$ and $w_i$. Note that by construction all of $u,w_k,w_i$ will be endpoints of maximal $(\kappa_1,\kappa_2)$-bichromatic paths. 
Consider the $(\kappa_1,\kappa_2)$-bichromatic path $P = p_1p_2 \dots p_t = w_k$ ending at $w_k$. If it starts at $w_i$, we consider the chain 
$$C = (uw_1, uw_2, \dots, uw_i) + (uw_i, p_1p_2, p_2p_3, \dots, p_{t-1}w_k)$$
This chain is shiftable, and shifting it leaves $\kappa_2$ as an available colour at both $p_{t-1}$ and $w_k$, and so we can extend the colouring by giving this edge the colour $\kappa_2$. See Figure~\ref{fig:i} for an illustration. 
\begin{figure}%
    \centering
    \subfloat{{\includegraphics[width=6cm]{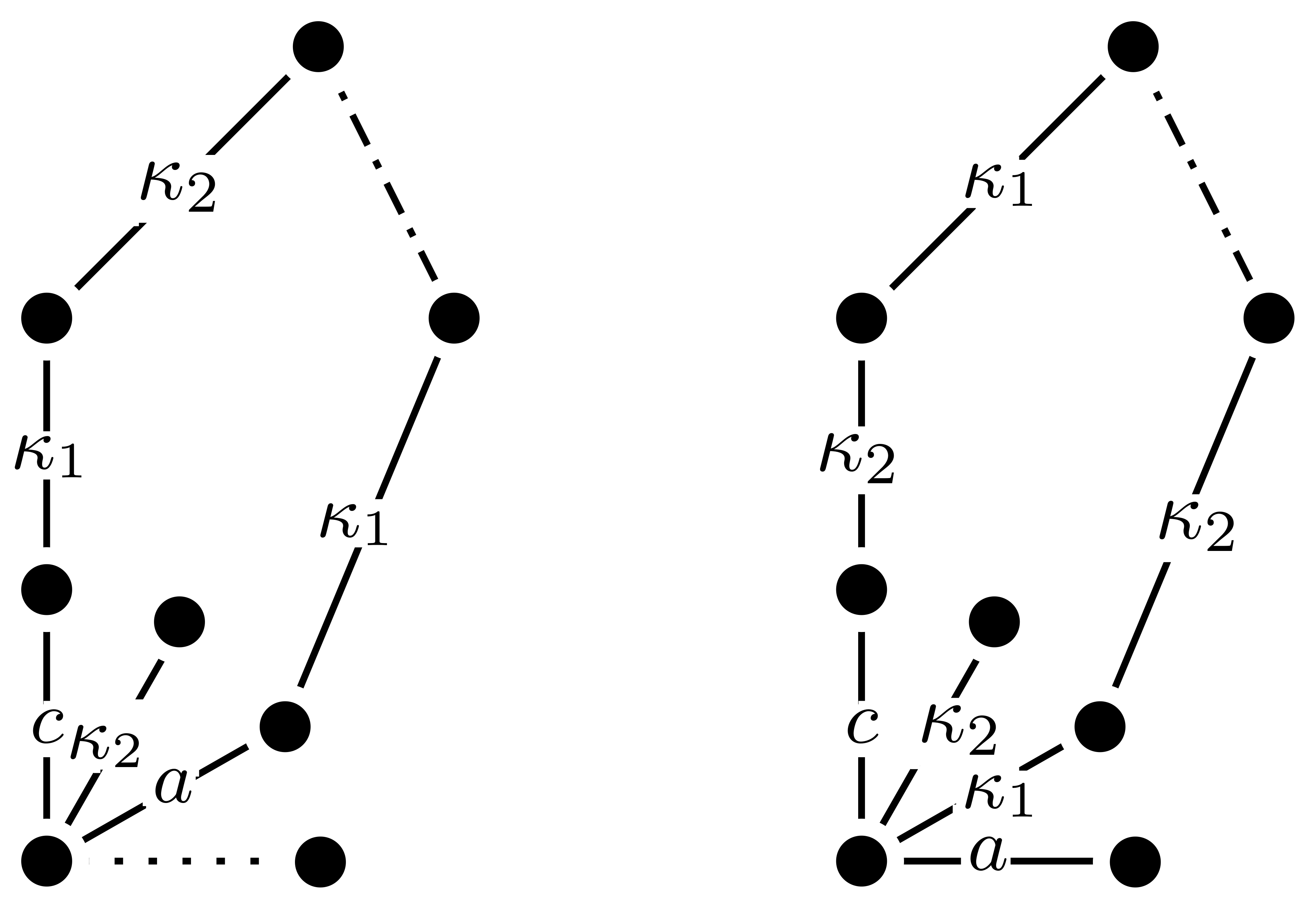} }}%
    \caption{An illustration of shifting the chain and extending the colouring.}
    \label{fig:i}%
\end{figure} 

If $P$ instead starts at $u$, we consider the chain 
$$C = (uw_1, uw_2, \dots, uw_{i+1}) + (uw_{i+1}, p_{1}p_{2}, p_{2}p_{3}, \dots, p_{t-1}w_k)$$ 
This chain is again shiftable, and shifting it again leaves $\kappa_2$ as an available colour at both $p_{t-1}$ and $w_k$, and so we reach the same conclusion as before. See Figure~\ref{fig:ip1} for an illustration. 
\begin{figure}%
    \centering
    \subfloat{{\includegraphics[width=6cm]{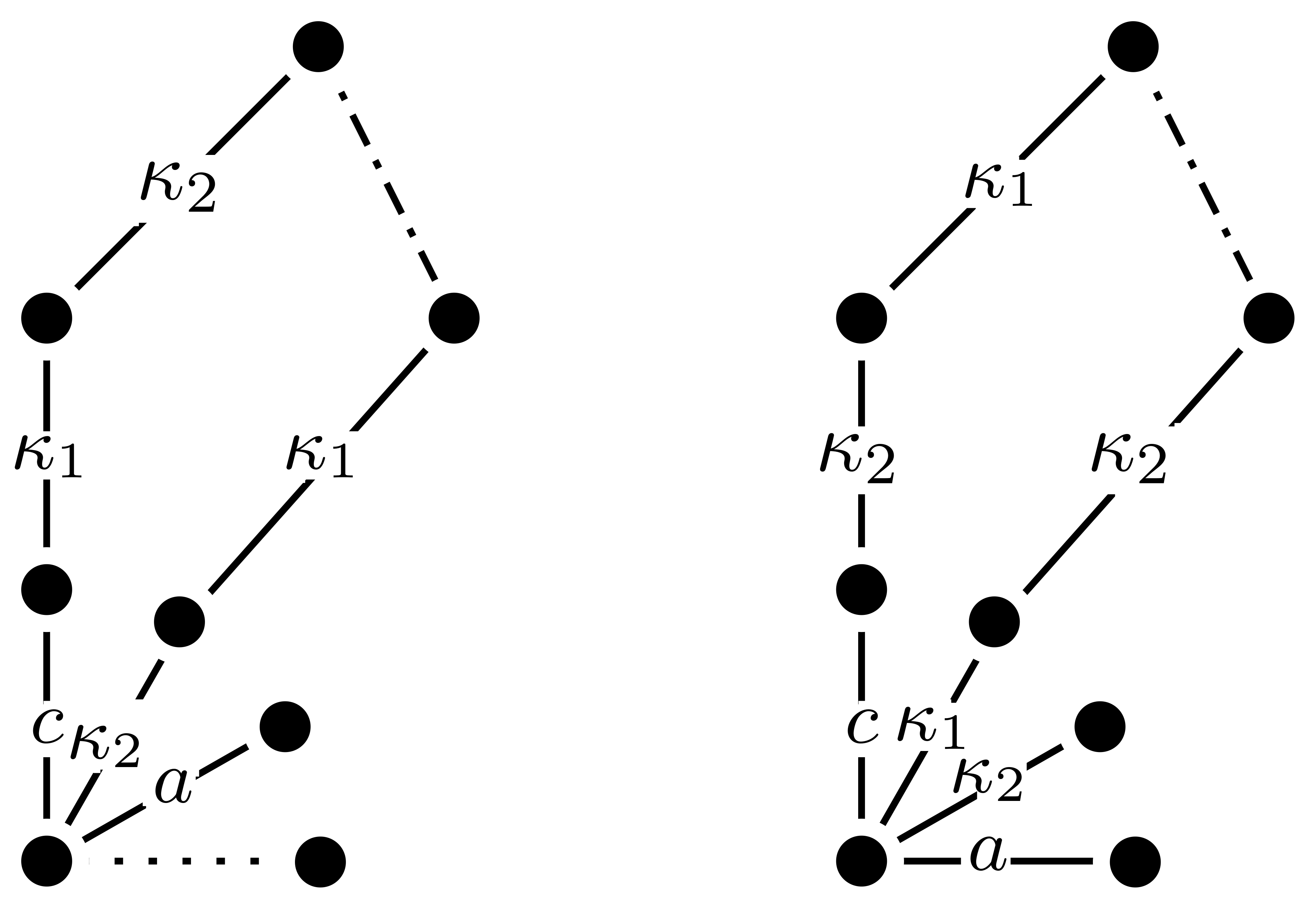} }}%
    \caption{An illustration of shifting the chain and extending the colouring.}
    \label{fig:ip1}%
\end{figure} 

If $P$ instead begins at any other vertex, we consider the chain 
$$C = (uw_1, uw_2, \dots, uw_k) + (uw_k, w_kp_{t-1}, p_{t-1}p_{t-2}, \dots, p_{2}p_{1})$$
It is again shiftable, and shifting it leaves either $\kappa_1$ or $\kappa_2$ as an available colour at both $p_{2}$ and $p_{1}$, and so we can extend the colouring to also colour this edge. %insert picture. 
\end{proof}

In some sense, Vizing proved a slightly stronger statement than the one that is usually referred to as Vizing's Theorem. He showed that given any proper partial edge-colouring with $\Delta + 1$ colours and any uncoloured edge $e$, one can always find an augmenting subgraph of size $O(\Delta + n)$ such that exchanging colours within this subgraph allows one to also properly colour $e$. 
The proof is constructive and it implies Vizing's Theorem as it gives rise to a polynomial-time algorithm for computing a $(\Delta + 1)$-edge-colouring. 
Note, however, that the subgraph is rather large; in order to colour a single extra edge, we might have to recolour $\Omega(n)$ edges. 
The subgraph constructed is what we refer to as a Vizing chain:
\begin{definition} \label{def:truncVC}
Given a proper partial colouring $c$ of $G$, a \emph{Vizing chain} of an uncoloured edge $e=uv$ (or a vertex $u$) is a $c$-shiftable chain $F + P$ consistsing of a fan chain $F = (uv, \dots, uw_{k})$ centered at $u$ together with a $(\kappa_1, \kappa_2)$-bichromatic path chain $P = w_{k}x_1x_2 \cdots x_p$ -- possibly of length 0 -- such that $\kappa_1$ is available at $u$ and $\kappa_2$ is available at $w_k$.
If $p\geq 1$ we say that $F+P$ \emph{ends} at the edge $x_{p-1}x_{p}$ and at the vertex $x_{p-1}$. 
Finally, we let $c' = \operatorname{Shift}(c, F+P)$ be the \emph{post-shift} colouring. 
\end{definition}

An important tool for proving the existence of smaller augmenting subgraphs is to glue together what we may call \emph{truncated Vizing chains} to form a \emph{multi-step Vizing chain}, a term coined by Bernshteyn~\cite{BERNSHTEYN}, but also considered in different shapes by other papers~\cite{duan,grebik2020measurable}. These papers all consider different ways of concatenating truncated Vizing chains in order to prove properties about edge-colourings. 
We consider a general definition of multi-step Vizing chains that encompasses all of these constructions as special types of Vizing chains. Later we will specialise to more restricted classes of multi-step Vizing chain.
In particular, we note that in Section~\ref{sec:app1}, we allow the Vizing chains to overlap in edges. Multi-step Vizing chains arise by combining Vizing chains:
\begin{definition}
Given a proper partial colouring $c$ of $G$, an $i$-step Vizing chain is a $c$-shiftable chain of the form $F_1+P_1 + F_2 + P_2 + \dots F_i + P_i$ where $F_j+P_j$ is a Vizing chain for all $j$.   
\end{definition}
%insert picture
In the case, where we do not specify $i$, we simply refer to such chains as a \emph{multi-step Vizing chain}. 
We stress again that we impose no restriction on whether the different chain components that make up the multi-step Vizing chain are allowed to overlap in edges, and we will consider chains that could potentially overlap in a lot of edges. 
A (multi-step) Vizing chain is \emph{augmenting} if shifting the chain leaves the final edge of the chain with a colour that is available at both endpoints of the edge. If the chain is not augmenting, it is \emph{truncated}. 

We let the \emph{length} of an $i$-step Vizing chain be the sum of the lengths of the paths $P_j$ that make up the multi-step Vizing chain, and the \emph{size} of an $i$ step Vizing chain be the total number of edges (with multiplicities) in the Vizing chain. Finally, we let $i$ be the \emph{chain length} of the multi-step Vizing chain.

\section{Types of multi-step Vizing chains} \label{sec:VC}
The goal of these multi-step Vizing chain constructions is to glue together mulitple truncated Vizing chains -- all satisfying some properties -- to end up with an augmenting Vizing chain that inherits some (possibly relaxed version) of these properties. 
Examples of properties could be that the chains should be short or that the pre- and post-shift colourings should adhere to certain restrictions on allowed colours for each edge. 

It is straight-forward to check that the construction used to prove Vizing's Theorem shows the existence of multi-step Vizing chains, but it is not really clear that they exist in any useful forms. Hence, we will briefly introduce some restricted classes of Vizing chains. 
Proving that some of these classes of Vizing chains contain augmenting multi-step Vizing chains will then imply our theorems.

\paragraph{Non-overlapping Vizing chains} This type of Vizing chain has been considered in different forms by~\cite{BERNSHTEYN,duan,grebik2020measurable}. We will construct these chains in a similarly to Bernshteyn~\cite{BERNSHTEYN}, who generalised the construction of Greb\'\i k and Pikhurko~\cite{grebik2020measurable}, but our analysis approach is completely different. First we define a \emph{non-overlapping} Vizing chain.
\begin{definition} \label{def:nonO}
An $i$-step Vizing chain is \emph{non-overlapping} if every pair of Vizing chains $F_j+P_j$ on edge $e_j$ and $F_k+P_k$ on edge $e_k$ share an edge exactly when $j = k-1$ and the shared edge is $e_k$.
Furthermore, for all $k \leq i-1$ no edge in the chain $F_k+P_k = (e_1, \dots, e_s)$ may be repeated i.e.\ $e_{j} = e_{\ell}$ iff $j = \ell$. 
\end{definition} 
It is straightforward to check that the construction used in the proof of Vizing's Theorem yields a $1$-step Vizing chain that is non-overlapping. However, it is not immediately clear how to extend it to a multi-step Vizing chain. To do so, we will apply the following lemma, due to Greb\'\i k and Pikhurko~\cite{grebik2020measurable}. 
%They used it to prove a measurable version of Vizing's theorem, and that proof, as noted by Bernshteyn, can be straightforwardly adjusted to show the existence of an augmenting two-step Vizing chain of size $O(\poly(\Delta)\sqrt{n})$. Later, it was used by Bernshteyn~\cite{BERNSHTEYN} to prove the existence of multi-step Vizing chains of size $O(\text{poly}(\Delta)\log^2 n)$. 
\begin{lemma}[\cite{BERNSHTEYN,grebik2020measurable}] \label{lma:fan2}
Let $c$ be a proper partial edge-colouring and let $e = uv$ be an uncoloured edge. Let $\kappa_{1} \in [\Delta+1]$ be an available colour at $u$ and let $\kappa_{2} \in [\Delta+1]$ be an available colour at $v$. Then there exists a fan $uw_{1}, \dots, uw_{k}$ with $w_1 = v$ such that either 
\begin{enumerate}
    \item $\kappa_2 \notin A(u)$ and the fan is maximal subject to the constraint that the colours $\kappa_1$ or $\kappa_2$ are considered unavailable to every vertex.
    \item $w_{k} \neq v$, the representative available colour at $w_{k}$ is $\kappa_2$ and no edge $uw_{i}$ is coloured $\kappa_2$.
    \item an available colour at $uw_{k}$ is also available at $u$. 
\end{enumerate}
\end{lemma}
\begin{proof}
We construct a maximal fan almost exactly as in the proof of Theorem~\ref{thm:VT}. Again we begin by setting $v = w_1$. If $\kappa_2 \in A(u)$, we are in case $3)$ and we can stop the construction.
Otherwise, having picked $w_j$, we show how to pick $w_{j+1}$ or conclude that we may stop the construction. First, we check if $A(w_{j})$ contains a colour -- that is \emph{not} $\kappa_1$ or $\kappa_2$ -- present at either $uw_{\ell}$ for some $\ell<j$ or at $u$. If this is the case, we may choose this colour as our representative available colour, and subject to the chosen representative colours our fan is now maximal and satisfies $1)$.
If no such colour exists, it is either because $\kappa_1 \in A(w_{j})$ or because $A(w_{j}) = \{ \kappa_2 \}$. In the former case, we may pick $\kappa_1$ as our representative available colour, and stop the construction since we are in case $3)$. In the latter case, we also stop the construction and note since $v = w_1$ is incident to at most $\Delta-1$ coloured edges, this cannot happen at $v$. Hence, we are in case $2)$.
\end{proof}

Now we can outline a generic way to construct multi-step Vizing chains. We will later show that this type of construction can be used to construct augmenting non-overlapping Vizing chains. 
The other constructions we consider will also be similar to (parts of) this construction.

\paragraph{Construction} Given an uncoloured edge $u_1v_1$, we construct a $1$-step Vizing chain similar to the proof-sketch of Theorem~\ref{thm:VT}. 
That is: we construct a maximal fan chain $M = (u_1w_{1,1}, \dots, u_1w_{1,k})$ such that $v_1 = w_{1,1}$, the representative available colour of $w_{1,j}$ is $c(u_{1}w_{1,j+1}) \in A(w_j)$ for all $j < k$, and that either $1)$ the representative available colour, $\kappa_2$, of $w_{1,k}$ is in $A(u_1)$ i.e.\ $\kappa_2 \in A(u_1)$ or $2)$ it is equal to the representative available colour $c(u_1w_{1,i+1})$ at the vertex $w_{1,i}$. 
Then we pick an arbitrary representative available colour, $\kappa_1$, for $u_1$ in $A(u_1)$, and consider the maximal $(\kappa_1,\kappa_2)$-bichromatic path $BP$ ending at $w_k$. 
In case $1)$, we let $F_1 = M$ and $P = u_1w_{1,k}$. It is straightforward to check that this yields an augmenting Vizing chain. 
In case $2)$, we have two options. We can either choose to truncate $BP$ and get a truncated Vizing chain, or we can choose to use the entire bichromatic path to get an augmenting Vizing chain. 
If we truncate $BP$, we set $P_1 = (u_1w_{1,k}, BP|t)$ for some $t < \operatorname{length}(BP)$. In this case, we set $F_1 = M$. Again it is straightforward to check that $F_1 + P_1$ is a Vizing chain. 

Otherwise, if we do not truncate $BP$, we have to choose $F_1$ accordingly: if $BP$ ends at a vertex that is not $w_{1,i}$ nor $u_1$, then we can set $F_1 = M$ and $P_1 = (u_1w_{1,k}, BP)$. 
Note that this will always be the case in Applications~\ref{sec:app2, sec:app3, sec:dynCol}, since here we do not allow $BP$ to overlap with the fan. 
For completeness, we describe how to construct the chains in the other cases anyways. 
Otherwise, if it ends at $u_1$, we set $F_1 = M|(i+1)$ and $P_1 = (u_1w_{1,i+1},R(BP))$, where $R(BP)$ is the reversal of $BP$ i.e.\ the edges considered in the opposite order. 
Finally, if it ends at $w_{1,i}$, we set $F_1 = M|i$ and $P_1 = (u_1w_{1,i},R(BP))$.

We can use Lemma~\ref{lma:fan2} to extend a (multi-step) Vizing chain: given an $i$-step Vizing chain ending at an edge $e = u_iv_i$ and a point $u_i$, we can extend this Vizing chain to an $(i+1)$-step Vizing chain by applying the lemma above to $e$ such that we either avoid or reuse the colours of the bichromatic path of the last Vizing chain used to get there. 
Namely, in case $1)$ we pick an available colour at $u$ that is not $\kappa_1$ or $\kappa_2$ to use in the construction of the bichromatic path. Such a colour has to exist since $u$ is incident to at most $\Delta-1$ coloured edges, and one of these edges has to be coloured $\kappa_2$.
In case $2)$, we pick $\kappa_1$ as the available colour at $u$ so that we get a $(\kappa_1, \kappa_2)$-bichromatic path. In case $3)$, we can construct an augmenting Vizing chain using a bichromatic path of length 0, and so we do this.
%Finally, we choose the fans and bichromatic paths similarly to above, when we construct the $(i+1)^{\textrm{th}}$ Vizing chain.

In this way, we can extend our $i$-step Vizing chain to an $(i+1)$-step Vizing chain such that the fan is (almost always) certain to never overlap the Vizing chain used to reach $u$. We will refer to the extension described above as \emph{extending the multi-step Vizing chain via Lemma~\ref{lma:fan2}}.
For a specific multi-step Vizing chain, we cannot guarantee that we have many choices of $u$ such that the $(i+1)^{\textrm{th}}$ bichromatic path does not overlap with the multi-step Vizing chain used to go there. However, when we consider many choices of multi-step Vizing chains at the same time, we can provide such a guarantee.

In Section~\ref{sec:app2}, we will show that for any uncoloured edge $e$ there exists an augmenting, non-overlapping multi-step Vizing chain of size at most $O(\poly(\Delta)\log n)$. As will be evident from the proof, if the chain is long we will have many choices for how to pick the augmenting Vizing chain. 
This allows us to consider multiple edges at once and choose augmenting Vizing chains that do not overlap, so that we may colour multiple edges in parallel. To formalise this choice, we introduce the notion of \emph{family-avoiding} Vizing chains.

\paragraph{Family-avoiding Vizing chains} The idea is that given some forbidden points $\mathbb{F}_i$, we would like to avoid extending our Vizing chain through a point in $\mathbb{F}_i$ when constructing the $(i+1)^{\textrm{th}}$ Vizing chain. In fact we will be a bit more strict. We will avoid extending our $(i+1)^{\textrm{th}}$ Vizing chain through any of the sets $\mathbb{F}_1, \dots, \mathbb{F}_{i}$. If this holds for all $i$, we say that the Vizing chain \emph{avoids} the family $\mathcal{F} = \{\mathbb{F}_i\}_{i = 1}^{\infty}$. Note that a Vizing chain may pass through points in $\mathbb{F}_j$ for all $j$, but it is only allowed to be truncated at points not in $\mathcal{F} = \{\mathbb{F}_i\}_{i = 1}^{k-1}$ at the $k^{\textrm{th}}$ step. More formally, we define a family-avoiding Vizing chain as follows:
\begin{definition} \label{def:famA}
Let $\mathcal{F}$ be a family of subsets of $V(G)$ i.e.\ $\mathcal{F} = \{\mathbb{F}_i\}_{i = 1}^{\infty}$ for $\mathbb{F}_i \subset V(G)$. A multi-step Vizing chain on $e$ is \emph{$\mathcal{F}$-avoiding} if the $i^{\textrm{th}}$ Vizing chain is not created on a vertex in $\bigcup \limits_{j = 1}^{i-1} \mathbb{F}_{j}$. 
\end{definition}
In Section~\ref{sec:app3}, we use family-avoiding Vizing chains to construct many vertex disjoint augmenting Vizing chains, so that e.g.\ a distributed algorithm can colour many edges in parallel. 

\paragraph{Local Vizing chains} For this class of Vizing chains, we want to put some restrictions on the post-shift colourings. For instance, we consider the following restricted class of Vizing chains:
\begin{definition} \label{def:striclyLocal}
We say a (multi-step) Vizing chain is \emph{strictly local} if the post-shift colouring is \emph{strictly local} meaning that it satisfies that every coloured edge $e = uv$ has a colour from $L(e) = \{1, 2, \dots \max\{d(u),d(v)\}+1\}$.
\end{definition}
Clearly, if we can show that for any edge $e$ left uncoloured by a strictly local colouring, there has to exist an augmenting strictly local multi-step Vizing chain on $e$ of finite size and step-length, then Theorem~\ref{thm:LVT} follows. 
Indeed, we may begin from an empty colouring and then colour each uncoloured edge using the existence of such a multi-step Vizing chain. 
We will show how to construct these chains in the next section.

In Section~\ref{sec:dynCol}, we will relax this notion of locality to construct chains which are $\varepsilon$-local. Before defining such chains, we define the set of \emph{$\varepsilon$-available} colours for a vertex $v$ as 
\[
A_{\varepsilon}(v) = A(v) \cap [(1+\varepsilon) d(v)]
\]
Now we can define $\varepsilon$-local Vizing chains.
\begin{definition} \label{def:epsLocal}
We say a (multi-step) Vizing chain is \emph{$\varepsilon$-local} if the following condition holds for all $\kappa$ in both the pre-shift colouring $c$ and the post-shift colouring $c'$ after shifting the entire (multi-step) Vizing chain:
\[
|\{e:\hat{c}(e) = \kappa \}| \leq |\{v:  \kappa \in [(1+\varepsilon)d(v)], \kappa \notin A_{\varepsilon}(v)\}|
\]
where $\hat{c} \in \{c, c'\}$.
\end{definition}
In Appendix~\ref{app:B}, we explain how we will construct such Vizing chains. 
In Section~\ref{sec:dynCol}, we show that the constructed chains indeed are $\varepsilon$-local and use them to construct a fully-dynamic explicit edge-colouring algorithm with worst-case guarantees. 

\paragraph{Known properties of Vizing chains}
Here, we translate known results on Vizing chains and multi-step Vizing chains to the terminology introduced above. 
Duan, He and Zhang~\cite{duan} showed that if one has $(1+\varepsilon)\Delta$ colours, then there exist augmenting non-overlapping multi-step Vizing chains of size $O(\poly(\Delta) \log^2 n)$ provided that $\Delta = \Omega(\varepsilon^{-2}\log^2 n)$.
In the setting where one is only allowed $\Delta + 1$ colours, the arguments of Greb\'\i k \& Pikhurko in~\cite{grebik2020measurable} essentially show the existence of augmenting non-overlapping 2-step Vizing chains of size $O(\poly(\Delta) \sqrt{n})$.
In~\cite{BERNSHTEYN} Bernshteyn generalised this approach to construct even shorter augmenting non-overlapping multi-step Vizing chains of size $O(\poly(\Delta) \log^2 n)$.
Finally, Chang, He, Li, Pettie and Uitto~\cite{pettie} showed a lower bound on the size of any augmenting subgraph of $\Omega(\Delta \log \frac{n}{\Delta})$. In particular, this lower bound must carry over to the size of any augmenting multi-step Vizing chain.

\section{Application 1: Local version of Vizing's theorem} \label{sec:app1}

We say an edge $e = uv$ is coloured \emph{strictly local} if $c(e) \in \{1, 2, 3, \dots, 1+\max\{d(u),d(v) \} \}$. A strictly local partial colouring is then a partial edge-colouring, where all coloured edges are coloured strictly local.
The Local Vizing Conjecture says that there exists a proper colouring of any graph where all edges are coloured strictly locally. 
We say a fan is \emph{strictly local} if the representative available colour at each point $w_i$ of the fan is in $\{1, 2, 3, \dots, 1+d(w_i)\}$.
Observe that if the original colouring is strictly local, then so is the colouring obtained from shifting a strictly local fan.
Strictly local fans gives us a way of constructing truncated Vizing chains that are strictly local, assuming that the pre-colouring is strictly local: fix a center $u$ and an uncoloured edge $uv$. The Vizing chain constructions from Lemma~\ref{lma:fan2} and Theorem~\ref{thm:VT} can be made strictly local by choosing the available colour at a vertex $y$ in $\{1, 2, 3, \dots, 1+d(y)\}$ and truncating $P$ at the first problematic edge. We will need the following generic construction of strictly local Vizing chains:

\paragraph{Construction:} Let $c$ be a strictly local partial colouring of a graph $G$, and let $e = uv \in E(G)$ be any edge left uncoloured by $c$. 
Then we consider the following strictly local Vizing chain on $e$:
we construct a maximal fan $uw_1, \dots, uw_k$ as follows. Initially, we set $w_1 = v$. Now having picked $w_j$, we pick $w_{j+1}$ or conclude our construction with $k = j$ as follows: Note first that for all vertices $y$ we have that $A(y) \cap [d(y)+1]$ is non-empty, since at most $d(y)$ colours can be excluded from $A(y)$. We now have three cases: \\\\
\textbf{Case 1:} if there exists a colour $\kappa \in A(w_j) \cap [d(w_j)+1]$ that is also available at $u$, we pick $\kappa$ as $w_j$'s representative strictly locally available colour and conclude the construction with $k = j$. \\\\
\textbf{Case 2:} if there exists a colour $\kappa \in A(w_j) \cap [d(w_j)+1]$ and an index $\ell<j$ such that $c(uw_{\ell}) = \kappa$, we pick $\kappa$ as $w_j$'s representative available colour and conclude the construction with $k = j$. \\\\
\textbf{Case 3:} Otherwise, we pick an arbitrary colour $\kappa \in A(w_j) \cap [d(w_j)+1]$ as $w_j$'s representative available colour. Note that $u$ has to be incident to an edge $ux$ coloured $\kappa$ and that $x \notin \{w_{\ell}\}_{\ell=1}^{j}$, since then we would be in case $1)$ or case $2)$. Therefore, we can pick $w_{j+1} = x$. \\\\
The above construction gives us a fan $F$ on the edges $uw_1, \dots, uw_{k}$, such that $w_k$'s representative available colour is either available at $u$ or is present at $u$ at the edge $uw_i$ for some $i<k$. Note that $\ell$-shifting $F$ leaves a colouring that is also strictly local, since $uw_i$ receives a colour in $[d(w_i)+1]$ by construction. In particular, if we stop in case $1)$, we have identified a strictly local Vizing chain, and we conclude the construction.

If we stop in case $2)$, we pick an arbitrary representative available colour $\kappa_1 \in A(u) \cap [d(u)+1]$ for $u$ and set $\kappa_2$ equal to $w_k$'s representative available colour. Then we consider the $(\kappa_1,\kappa_2)$-bichromatic path $P' = p_1, p_2, \dots, p_{t} = w_k$ ending at $w_k$. We choose $P$ to be the union of a fan edge and a subpath of $P'$ as follows: beginning from $w_k$, we traverse $P'$. We choose $P$ by truncating $P'$ at the first edge $xy$, we meet along $P'$ such that shifting the colours of $P'$ would cause the colour of $xy$ to violate strict locality of the post-shift colouring. 
This concludes the construction in case $2)$. 

Finally, if no such edge exist, we can construct an augmenting strictly local Vizing chain in a similar fashion to what we did in the proof of Theorem~\ref{thm:VT}. Currently, we have a maximal strictly local fan consisting of vertices $u,w_{1}, \dots, w_{k}$ and the representative free colour at $w_k$ also the representative free colour at $w_i$ for some $i < k$. 

Now let $\kappa_1$ be a free colour at $u$ and $\kappa_2$ the representative free colour at $w_k$ and $w_i$. Note that by construction all of $u,w_k,w_i$ will be endpoints of maximal $(\kappa_1,\kappa_2)$-bichromatic paths. 
Again we consider the $(\kappa_1,\kappa_2)$-bichromatic path $P' = p_1p_2 \dots p_t = w_k$ ending at $w_k$. If it starts at $w_i$, we consider the chain 
$$C = (uw_1, uw_2, \dots, uw_i) + (uw_i, p_1p_2, p_2p_3, \dots, p_{t-1}w_k)$$
This chain is shiftable, and shifting it leaves $\kappa_2$ as an available colour at both $p_{t-1}$ and $w_k$, and so we can extend the colouring by giving this edge the colour $\kappa_2$.
We note the following two things: 1) if the original colouring is strictly local, then so is the one obtained by shifting the above chain, and 2) by assumption colouring the last edge like this, does not violate strict locality, since otherwise we would have truncated at this edge.

If $P$ instead starts at $u$, we instead consider the chain 
$$C = (uw_1, uw_2, \dots, uw_{i+1}) + (uw_{i+1}, p_1p_2, p_2p_3, \dots, p_{t-1}w_k)$$ 
This chain is again shiftable, and shifting it again leaves $\kappa_2$ as an available colour at both $p_{t-1}$ and $w_k$, and so we reach the same conclusion as before.
Again we know that shifting the second part of $C$ does not inviolate the strict locality of the colouring. 

If $P$ instead begins at any other vertex, we consider the chain 
$$C = (uw_1, uw_2, \dots, uw_k) + (uw_k, w_kp_{t-1}, p_{t-1}p_{t-2}, \dots, p_{2}p_{1})$$
It is again shiftable, and shifting it leaves either $\kappa_1$ or $\kappa_2$ as an available colour at both $p_{2}$ and $p_{1}$, and so we can extend the colouring to also colour this edge without violating strict locality.

Note that if we end up with a truncated Vizing chain, we may recursively perform the exact same construction at the truncated edges to construct a multi-step Vizing chain on $e$. In particular, observe that these Vizing chains might overlap. In this section, we will only consider strictly local (multi-step) Vizing chains constructed as explained above. 

Consider a strictly local Vizing chain constructed as explained above, and let $e'$ be the the last edge of $P$ i.e.\ the truncated edge. A simple but key observation is that $e'$ has to go between two vertices of (relative) low degree.
\begin{observation} \label{obs:SL}
Let $c$ be a strictly local (partial) colouring of a graph $G$, and let $P$ be any $(\kappa_1, \kappa_2)$-bichromatic path. If the shift of $P$ makes $e = xy$ violate the strict locality, then $\max\{d(x), d(y) \} +1 < \max \{\kappa_1, \kappa_2\}$.
\end{observation}
\begin{proof}
If the stated inequality was not satisfied, then $e$ would not violate strict locality of the post-shift colouring.
\end{proof}
The main lemma of this section will show that if we construct a multi-step Vizing chain by repeatedly applying the construction from above on the truncated edges, then this construction cannot continue indefinitely before it ends up in an augmenting Vizing chain. 

To prove this lemma, we consider the following potential function defined for any graph $G$ and any partial colouring $c: G \mapsto [\Delta + 1] \cup \{\neg\}$ of $G$:
\begin{align*}
\Phi(G,c) &= \sum \limits_{\kappa = 1}^{\Delta + 1} |\{v: \kappa \in A(v) \cap [1+d(v)] \}| \\
&\leq n (\Delta+1)
\end{align*}

We will show that each time one shifts one of the strictly local truncated Vizing chains constructed above, the potential of the post-shift colouring will be at least one smaller than that of the pre-shift colouring.
Since for any choice of $G,c$ we have $\Phi(G,c) \geq 0$, it follows that there has to exist an augmenting and strictly local multi-step Vizing chain of length at most $O((\Delta+n) \cdot{} n (\Delta+1))$. 

We first show that shifting a strictly local fan never increases the potential:
\begin{proposition} \label{prop:fanShift}
Let $c_1,c_2$ be two (partial) colourings of the graph $G$ such that $c_2$ is obtained from $c_1$ by $i$-shifting a strictly local fan. Then $\Phi(G,c_2) \leq \Phi(G,c_1)$.
\end{proposition}
\begin{proof}
Let $uw_1, \dots, uw_i$ be the first $i$ edges of the shifted strictly local fan with center $u$.
Note that after shifting $uw_1, \dots, uw_i$, the contribution of $u$ to $\Phi$ will not change, as the same colours are incident to $u$. The new available colours at $w_2, \dots, w_i$ contribute at most $i-1$ extra to the potential after the shift, but -- by construction -- the now no-longer available colours at $w_1, \dots, w_{i-1}$ drop the potential by exactly $i-1$ after the shift.
\end{proof}
We then show that if $F+P$ is a strictly local 1-step Vizing chain constructed as above, then the following restrictions on $P$ holds:
\begin{proposition} \label{prop:pathEP}
Let $F+P$ be a strictly local truncated Vizing chain constructed as above. Suppose $F$ contains the edges $uw_0, \dots, uw_k$ and that the colour chosen as available at $w_k$ was also chosen as available for $w_i$. 
Then the last edge of $P$ cannot have $u, w_k$ or $w_{i}$ as an endpoint. 
\end{proposition}
\begin{proof}
Suppose that $P$ is $(\kappa_1, \kappa_2)$-bichromatic and that $\kappa_1$ was chosen as the representative available colour at $u$ and $\kappa_2$ as the representative available colour at $w_k$. 
If the last vertex of $P$ was $u$, then the last edge of $P$ would have to be $uw_{i+1}$ and it would have to be coloured $\kappa_2$. Since $\kappa_1$ is available at $u$ by construction, we reach a contradiction with Observation~\ref{obs:SL}.

The last vertex of $P$ can never be $w_k$, as in this case $P$ would be augmenting and $\kappa_1 = \kappa_2$. If $w_k$ is the second last vertex, we also reach a contradiction, since the the last edge of $P$ must have been coloured $\kappa_1$ initially, and hence it may be coloured $\kappa_2$ afterwards, which is available at $w_k$ by assumption.

Finally suppose that the last edge of $P$ has $w_i$ as an endpoint. Then since $\kappa_2$ was chosen as an available colour at $w_i$, it follows that the path has to end at $w_i$ and, furthermore, that it can never be problematic that the last edge changes its colour to $\kappa_2$, so this is also a contradiction. 
\end{proof}
We are now ready to prove the main lemma of this section, which says that the potential has to drop, when shifting this particular kind of strictly local truncated Vizing chain:
\begin{lemma} \label{lma:mainLVT}
Let $F+P$ be a strictly local truncated Vizing chain constructed as before. Then shifting $F+P$ to go from the pre-colouring $c_1$ to the post-colouring $c_2$ drops the potential by at least 1 i.e.\ $\Phi(G,c_1) \geq \Phi(G,c_2) + 1$
\end{lemma}
\begin{proof}
Suppose that $P$ is $(\kappa_1,\kappa_2)$-bichromatic, and suppose that $\kappa_1$ was chosen as available at $u$. Observe first that the contribution of $u$ and $w_k$ will drop by $2$ after the shift of $F+P$. 
Indeed, after the shift $u$ will now be incident to an edge coloured $\kappa_1 \in [d(u)+1]$, and $w_k$ will now be incident to an edge coloured $\kappa_2 \in [d(w_k)+1]$. 
By Proposition~\ref{prop:fanShift} and Proposition~\ref{prop:pathEP} neither $u$ nor $w_k$ lose any colours incident to them, and hence their contribution to the potential has to decrease by at least $2$. 

Observe secondly that if $xy$ is the last edge of $P$, then $c(xy) = \min \{\kappa_1, \kappa_2 \}$ and so by Observation~\ref{obs:SL} the contribution of $x$ is the same before and after the shift, and the contribution of $y$ drops by at most $1$. 

Finally, note that any interior vertex of $P$ sees the same colours before and after the shift, and so their contribution to the potential is unchanged. Hence, summing up the changes in contribution yields:
\[
\Phi(G,c_2) \leq \Phi(G,c_1) - 2 + 1 \leq \Phi(G,c_1) - 1
\]
\end{proof}
Now we may prove the local version of Vizing's Theorem:
\begin{proof}[Proof of Theorem~\ref{thm:LVT}]
Suppose we are given a strictly local partial colouring of $G$ with a minimum number of uncoloured edges. If no edge is uncoloured, we are done, so we may assume this is not the case and pick an arbitrary uncoloured edge $e = uv$. 
Next, we build a strictly local multi-step Vizing chain on $e$ as described earlier: construct a strictly local maximal fan on $e$ with edges $uw_1, \dots, uw_k$. If it is augmenting, we are done, so assume it is not. 
Then we extend by a subpath of a bichromatic path alternating in suitably chosen colours as explained earlier. Finally, we truncate the Vizing chain if we meet an edge that would invalidate strict locality after the shift. 
If this happens, we repeat the same process to extend to a multi-step Vizing chain of one further step-length. 
The process is repeated on the post-shift colouring of the multi-step Vizing chain that we wish to extend which is not necessarily identical to the original colouring.
By Lemma~\ref{lma:mainLVT} the potential of the post-shift colouring drops by at least one everytime we perform this construction, and so we create an augmenting strictly local Vizing chain after at most $O(n\Delta)$ steps. 
By augmenting this Vizing chain, we construct a larger strictly local (partial) colouring -- a contradiction.
\end{proof}
Note that this approach is constructive and also gives an algorithm for computing a local Vizing colouring. 
Observe that increasing the size of the colouring never increases the potential, and so we have to create and shift $O(m+n\Delta)$ Vizing chains. 
Each such chain can be shifted in $O(\Delta + n)$ time, and so we arrive at Corollary~\ref{cor:algoLVT}, which we repeat below for the readers' convenience:
\begin{corollary}[Corollary~\ref{cor:algoLVT}]
There is an an algorithm that runs in $O(n^2 \Delta)$-time for computing an edge-colouring such that every edge receives a colour from the list $L(uv) = [1+\max \{d(u),d(v)\}]$. 
\end{corollary}

%\paragraph{Faster strictly local colouring}
%We briefly remark that if $\Delta$ is small we may use the colour procedure from Section~\ref{sec:dynCol} to speed up the algorithm for computing a local Vizing colouring. 
%The idea is to modify the recolouring procedure described in Section~\ref{sec:dynCol} to choose colours in a strictly local way as described above.
%The construction is stopped, if the chain found is augmenting, or if one locates an edge which we cannot shift without breaking the strict locality of the post-shift colouring. 
%One only needs to observe that if the fans are constructed in a strictly local manner, then the Vizing chains constructed above will be strictly local, since truncating at an arbitrary edge at most increases the potential by $2$ so by the proof of Lemma~\ref{lma:mainLVT}, we know that the potential will never increase. Hence, we may construct and shift the required chains in $O(\poly(\Delta) \log^3 n)$ time w.h.p.\ per shift and so we get an algorithm running in $O(\poly(\Delta)\cdot{}n\log^3 n)$ time with high probability. 

\section{Application 2: Small augmenting subgraphs} \label{sec:app2}
The basic observation underlying our argument is the following: the average length bichromatic paths are short. In fact, it is sufficient to consider an even more basic fact: a graph with maximum degree $\Delta$ and $n$ vrtices can contain at most $n\Delta/2$ edges. In particular this means that we can employ a density based argument: if we want to pack a lot of long bichromatic paths into our graph, then the graph has to -- simply put -- be big. To that end we try to construct a short Vizing chain similarly to Bernshteyn~\cite{BERNSHTEYN}. But where Bernshteyn uses a probabilistic analysis to show that with non-zero probability we find a short augmenting Vizing chain, we will instead apply a density based argument, and show that if we do not succeed in finding a short Vizing chain of step-length $i$, then if we truncate and grow a new set of bichromatic paths, we will reach an exponentially increasing amount of new vertices. 
This allows us to get rid of the $\log n$ term in the length of each of the bichromatic paths, and hence we will get Vizing chains of size $O(\poly(\Delta)\log{n})$.

In this section, we will prove Theorem~\ref{thm:LOCALVT}, which we restate below for convenience.
\begin{theorem}[Theorem~\ref{thm:LOCALVT}]
Given a proper partial edge-colouring of a graph as well as an uncoloured edge $e$, there exists an augmenting multi-step Vizing chain of $e$ of size at most $O(\Delta^7 \log n)$.
\end{theorem}
The idea is to first construct an augmenting Vizing chain of $e$ -- we know one exists by Vizing's Theorem. If it is too long, we will only consider the first $\ell$ vertices along the bichromatic path of the Vizing chain. 
These vertices form the set $R_{1}$, i.e.\ it is the set of vertices reachable by a $1$-step Vizing chain from $u$. 
To each vertex of $R_1$, we may associate a unique edge among the first $\ell$ edges of the bichromatic path, such that every vertex is incident to the associated edge (simply associate the $i^{\textrm{th}}$ edge to the $i^{\textrm{th}}$ vertex). 
Now we will deem some of these edges \emph{unfit} to be the basis for a new Vizing chain. In particular, for each associated edge consider the act of shifting the Vizing chain ending at that edge, and then applying Lemma~\ref{lma:fan2} and the following discussion to construct a Vizing chain extension on that edge. 
Since we later will count how many such Vizing chains an edge can be a part of, we only wish to shift the colour of an edge once. 
Thus we will say a point or its associated edge is unfit for shifting if the resulting Vizing chain contains edges that are also included in the multi-step Vizing chain used to reach the edge (in fact we will be slightly stricter). By removing unfit edges -- and by extension their associated vertices -- from contention, we end up with a new set of vertices $N_{1}$. Now we are safe to begin new Vizing chains from these vertices, and so we will do exactly that. The new reachable vertices then form $R_2$, and we will repeat this construction until we find a short, augmenting and non-overlapping Vizing chain. 
The key property is that the size of $N_{i}$ will grow exponentially -- even if we only set $\ell = O(\poly(\Delta))$. The reason for this is that on average when we construct a new Vizing chain, most of the cut-off points will not be unfit.

The proof will proceed by induction, where we will show that given a properly constructed set $N_{i-1}$, we must have that $N_{i}$ is at least a constant factor larger. First we will give some formal definitions. 
\begin{definition} \label{def:Ri}
Given a proper partial colouring and an uncoloured edge $e$, we let $v \in R_i(e,\ell)$ denote the set of vertices such that there exists a non-overlapping $i$-step Vizing chain from $e$ ending at $v$ such that every Vizing chain in the $i$-step Vizing chain has length at most $\ell$.
\end{definition}
\begin{definition} \label{def:Ni}
Given a proper partial colouring and an uncoloured edge $e$, we let $v \in N_i(e,\ell) \subset R_i(e,\ell)$ if there exists a non-overlapping $i$-step chain from $e$ ending at $v$, and, furthermore, the multi-step Vizing chain may be extended to a non-overlapping $(i+1)$-Vizing chain by concatenating a Vizing chain of length at most $\ell$ at $v$.
\end{definition}
%Note to self: whether you allow R_i and R_j to overlap does not really matter. In practice, you could allow them not to overlap since we remove all the overlapping points any way. 
Often we will suppress the arguments $e$ and $\ell$ if their values are clear from the context. It is also important to note that even though $v \in N_{i}$ may be reached via multiple $i$-step Vizing chains or may be extended by different $i$-step Vizing chains, we will always construct only one specific extension for $v$. Before we state the main lemma of this section, we present an observation and some lemmas, which we will use to control our density-based argument.

\subsection{Controlling non-overlapping Vizing chains}
We will use the following simple upper bound on the density of a graph with maximum degree $\Delta$:
\begin{observation} \label{obs:density}
Let $G$ be a graph with maximum degree $\Delta$, and let $H$ be any subgraph of $G$. Then $\frac{|E(H)|}{|V(H)|} \leq \frac{\Delta}{2}$
\end{observation}
\begin{proof}
Observe that $2|E(H)| = \sum \limits_{v \in H} d(v) \leq \Delta |V(H)|$ by the handshaking lemma. 
\end{proof}
In order to control the expansion of $R_i$, we will use the following lemma which shows that one cannot pack many Vizing chains into a subgraph containing few edges:
\begin{lemma} \label{lma:packingLemma}
Let $S$ be a set of non-overlapping Vizing chains in a graph $G$ of maximum degree $\Delta$ such that each Vizing chain is constructed on a different vertex. Then any coloured edge in $G$ belongs to at most $4\Delta^2$ of the Vizing chains.
\end{lemma}
\begin{proof}
Let $e$ be an arbitrary coloured edge. Note first that $e$ belongs to the fan of at most $2$ of the Vizing chains in $S$, since we have at most one Vizing chain from $S$ on each point, and $e$ can only belong to the fan of a Vizing chain constructed on one of its endpoints. 

Note secondly that $e$ is part of at most $\Delta$ bichromatic paths -- one for each choice of second colour. Each bichromatic path is associated to at most $2\Delta$ Vizing chains, since the bichromatic path can only be used as a path component of a Vizing chain constructed at a neighbour of one of the endpoints of the bichromatic path.

Summing the above shows that each edge is part of at most $4\Delta^2$ Vizing chains from $S$. 
\end{proof}
Later we will also need to control the number of Vizing chain that intersects a vertex in a subgraph $H$. 
In order to not have to worry about exactly which fans -- and hence Vizing chains -- that are constructible, we will upper bound the number of points that have a neighbour that is the endpoint of a maximal bichromatic path that reaches $H$ instead. 
Regardless of which fans are actually realizable, any Vizing chain beginning at a vertex $v \in G-N^{1}(H)$ cannot reach $H$ unless $v$ has such a neighbour. 
Hence this upper bound also holds for the number of vertices in $G-N^{1}(H)$ that have a Vizing chain reaching $H$.
An argument with a similar flavour to the one above shows that:
\begin{lemma} \label{lma:reachingLemma}
Let $H \subset G$ be a subgraph of $G$ and let $c,c'$ be two (partial) $(\Delta+1) $edge-colouring of $G$ differing only on $H$. 
Then, under both $c$ and $c'$, the same at most $2(\Delta+1)^3 \cdot{} |H|$ points in $G-N^{1}(H)$ can reach $H$ via a maximal bichromatic path located at a neighbour $y$.
Here $N^{1}(H)$ denotes the one-hop neighbourhood of $H$ i.e.\ the vertices in $G$ that has distance at most one from $H$ in $G$.
\end{lemma}
\begin{proof}
We first note that since $c$ and $c'$ only differ on $H$, any maximal bichromatic path located at a neighbour $y$ reaching $H$ looks the same in both $c$ and $c'$ until it reaches the first vertex in $H$, and so the vertices in $G - N^{1}(H)$ reaching $H$ via a maximal bichromatic path located at a neighbour $y$ under $c$ and $c'$ are identical.

Hence we only need to bound the number of vertices in $G-N^{1}(H)$ that reach $H$ via a maximal bichromatic path located at a neighbour $y$ under $c$. 
%Only fans in the one-hop neighbourhood of $H$ can reach $H$, so we only have to count the number of Vizing chains whose bichromatic paths reach $H$. 

For a vertex $v$ in $H$, we note that it belongs to at most $(\Delta+1)^2$ bichromatic paths in $G-H$ that intersect $H$ for the first time in $v$. One for each pair of colours. 
Each bichromatic path can begins at a neighbour of at most $2\Delta$ vertices -- namely the neighbours of the endpoints of the bichromatic paths. 
Hence, in total we have at most $2(\Delta + 1)^3$ such vertices for each vertex in $H$, and so the lemma follows.
\end{proof}

In the next section, we will consider many different $(i+1)$-multi-step Vizing chains at once, and we would like to use Lemma~\ref{lma:reachingLemma} to argue that not too many of these multi-step Vizing chains can be overlapping at the same time.
In order to do so, we follow a two step approach. First we note that when we extend an $i$-step Vizing chain via Lemma~\ref{lma:fan2} very few extension points will extend via a Vizing chain that intersects the $i^{\textrm{th}}$ Vizing chain in a vertex in the fan or in an edge of the bichromatic path. In particular, if the $(i+1)^{\textrm{th}}$ Vizing chain intersects the $i^{\textrm{th}}$ Vizing chain, then it intersects $P_i$ in a vertex and not an edge. This is due to the way the colours of $P_{i+1}$ is chosen in Lemma~\ref{lma:fan2}.

Hence, we can limit ourselves to consider extensions through points where the $(i+1)^{\textrm{th}}$ Vizing chain does not overlap with the $i^{\textrm{th}}$ Vizing chain in any edges (but possibly in some vertices).
Hence, the challenge is to argue that if we consider many such extension points, not too many of them can intersect the first $(i-1)$ steps of the chains used to reach them. 
In order to argue this, we would like to be able to apply Lemma~\ref{lma:reachingLemma} with $H$ equal to the union of the first $(i-1)$ steps of all of the multi-step Vizing chains, we consider. 
The idea is that any extension that intersects $H$, has to be constructed on a point that has a neighbour $y$ that can reach $H$ via a maximal bichromatic path under both the $i$-shifted colouring $c'$ used to reach the extension point and the original colouring $c$.
Note here that $c'$ will differ for different extension points. 
This would then allow us to only reason about the original colouring $c$ as we do in the proof of Lemma~\ref{lma:reachingLemma}. 
However, the colourings we are considering are actually the ones obtained by shifting the first $i$ steps of the multi-step Vizing chain, and so each of these colourings might differ from $c$ not only on $H$, but also along the edges in the $i^{\textrm{th}}$ step. 
Since we ensure that we only consider extensions that are edge-disjoint from the $i^{\textrm{th}}$ step, any Vizing chain extension that intersects $H$ i.e.\ the first $(i-1)$ steps of the chain, has to have a neighbour $y$ that is an endpoint of a maximal bichromatic path that reaches $H$ in the original colouring $c$. 
To see this observe that by Lemma~\ref{lma:reachingLemma} shifting the first $i-1$ steps of the chains does not create new extension points with such a neighbour $y$ and finally that shifting the $i^{\textrm{th}}$ chain might create new choices of such neighbours $y$, but that these new choices of $y$ must reach $H$ via maximal bichromatic paths that overlap with the edges of the $i^{\textrm{th}}$ Vizing chain. 
Indeed, otherwise these bichromatic paths also exist in $c$. 
Hence, we can apply the bound from Lemma~\ref{lma:reachingLemma} to count the number of Vizing chain extensions that might intersect the first $(i-1)$ steps. 
From now on, we will be slightly imprecise and do this by saying that we apply Lemma~\ref{lma:reachingLemma}, even though we technically only apply it implicitly via the above discussion. 

Finally, we note that we do consider Vizing chains that go through vertices in $N^{1}(H)-H$, but that we never extend through points in $N^{1}(H)$. This is done to avoid the following technicality: it is unclear whether Vizing chains build on a vertex in $N^{1}(H)-H$ goes through $H$ or not. By not extending via these vertices, we avoid having to deal with this.

\subsection{Existence of a small augmenting subgraph}
We are now ready to show the following main lemma, which we will use to prove Theorem~\ref{thm:LOCALVT} by induction.
\begin{lemma} \label{lma:mainApp2}
Suppose we are given a proper partial colouring and an uncoloured edge $e$ in a graph of maximum degree $\Delta$, and let $\ell \geq 80(\Delta+1)^7$.
Suppose furthermore that $\bigcup \limits_{j}^{i} N_j(e,\ell)$ contains no vertices contained in an augmenting $i$-step Vizing chain with step length $\ell$, and that for all $1 \leq j \leq i-1$ the following holds
\begin{enumerate}
    \item $|N_{j}(e,\ell)| \geq \frac{\ell}{8\Delta^3}\sum\limits_{k=1}^{j-1} |N_k(e,\ell)|$
    \item $|N_j(e,\ell)| \geq |R_j(e,\ell)|/2$.
\end{enumerate}
Then condition $1)$ and $2)$ also hold for $j = i$.
\end{lemma}
We first show that if the conditions of Lemma~\ref{lma:mainApp2} are satisfied then $|R_i(e,\ell)|$ is big.
\begin{lemma} \label{lma:boundR}
Suppose the conditions of Lemma~\ref{lma:mainApp2} hold. Then 
\[
|R_i| \geq \paren{\frac{\ell}{2\Delta^3}-1}|N_{i-1}| \geq \frac{1}{4} \paren{\frac{\ell}{\Delta^3}} \sum\limits_{k=1}^{i-1} |N_k|
\]
\end{lemma}
\begin{proof}
By assumption, we may extend exactly one $(i-1)$-step Vizing chain through every vertex in $N_{i-1}$ to get non-overlapping $i$-step Vizing chains, where the length of the $i^{\textrm{th}}$ Vizing chain is at least $\ell$.
Consider the subgraph $H$ spanned by the $\ell$ first edges of each bichromatic path in the $i^{\textrm{th}}$ Vizing chains concatenated at the vertices in $N_{i-1}$. By definition of $R_i$, we have
\[
|V(H)| \leq  |N_{i-1}| +  |R_{i}|
\]
Furthermore, since we only construct one Vizing chain from every point in $N_{i-1}$, an edge in $H$ belongs to at most $4\Delta^2$ different Vizing chains by Lemma~\ref{lma:packingLemma}. 
Since we have assumed that each Vizing chain has length at least $\ell$, we know that each bichromatic path contributes $\ell$ edges to $H$, and therefore we find:
\[
4\Delta^2 |E(H)| \geq \ell \cdot{} |N_{i-1}|
\]
Hence it follows by Observation~\ref{obs:density} that
\[
\frac{\ell|N_{i-1}|}{4\Delta^2(|N_{i-1}| +  |R_{i}|)} \leq \frac{|E(H)|}{ |V(H)|} \leq \frac{\Delta}{2}
\]
By rearranging, we conclude that
\[
|R_i|  \geq \paren{\frac{\ell}{2\Delta^3}-1}|N_{i-1}|
\]
By assumption, we have $|N_{i-1}| \geq \frac{\ell}{8\Delta^3}\sum \limits_{k = 1}^{i-2} |N_{k}|$, so since $\ell \geq 80\Delta^3$ is large enough we find
\[
|N_{i-1}| \geq 10 \sum \limits_{k=1}^{i-2} |N_k|
\]
so 
\[
\sum \limits_{k=1}^{i-1} |N_k| \leq \frac{|N_{i-1}|}{10}+|N_{i-1}|
\]
and it follows that
\[
\frac{10}{11} \sum \limits_{k=1}^{i-1} |N_k| \leq |N_{i-1}|
\]
Hence we find
\[
|R_i| \geq \paren{\frac{\ell}{2\Delta^3}-1}|N_{i-1}| \geq \frac{10}{11} \paren{\frac{\ell}{2\Delta^3}-1}\sum\limits_{k=1}^{i-1} |N_k| \geq \frac{1}{4} \paren{\frac{\ell}{\Delta^3}} \sum\limits_{k=1}^{i-1} |N_k|
\]
\end{proof}
Next we conclude that this implies that $|N_{i}|$ also is big. To do so, we show that only a constant fraction of the vertices in $|R_i(e,\ell)|$ are extended by overlapping Vizing chains.
\begin{lemma} \label{lma:boundN}
Let the assumptions be as in Lemma~\ref{lma:mainApp2}. Then at most
\[
3(\Delta+1)^4 \sum \limits_{k = 0}^{i-2} |R_k|+(3(\Delta+1)^4+2\Delta)|N_{i-1}|
\]
points in $R_{i}$ has no non-overlapping Vizing chain extension.
\end{lemma}
\begin{proof}
Observe first that if we can guarantee that the Vizing chain beginning at a point in $R_{i}$ will not intersect a vertex in the one-hop neighbourhood of $\bigcup \limits_{k = 0}^{i-1} R_k$, then surely it is non-overlapping. We will show something a little weaker: namely that we can pick the Vizing chains so that most will be vertex disjoint from the one-hop neighbourhood of $\bigcup \limits_{k = 0}^{i-2} R_k$, and furthermore they will not overlap the last Vizing chain, used to reach them, in any edges. % Do they have to be vertex disjoint through the one-hop neighbourhood of $R_k$? No, but we have to remove the one-hop neighbourhood of $R_k$ since the fans situated at these vertices might contain vertices in $R_k$. Problem is, we have fans on top of everything in there so let us do it to be sure
Finally, we may simply remove the 1-hop neighbourhood of vertices in $N_{i-1}$, since we do not wish to truncate our Vizing chains at these points, as we cannot guarantee that the fans constructed at these points are non-overlapping. 
%Since we remove the 1-hop neighbourhoods of any previously visited points, we are certain that the constructed fans does not overlap with any edges of fans constructed at $\bigcup \limits_{k = 0}^{i-2} R_k$
First of all, since we always use Lemma~\ref{lma:fan2} to extend a multi-step Vizing chain, we can count the number of vertices that will extend their Vizing chain so that it overlaps the last Vizing chain used to reach it. 
This can only happen in two ways: either the Vizing chain goes through the fan situated at a vertex in $N_{i-1}$ or it attempts to pick exactly the same augmenting path used by the last Vizing chain to reach it. There are at most $3(\Delta+1)^4 |N_{i-1}|$ chains of the first kind, and at most $2\Delta|N_{i-1}|$ of the second kind. 
Indeed, the one-hop neighbourhood of $N_{i-1}$ has size at most $(\Delta+1)|N_{i-1}|$, and hence by a similar argument to above at most $(\Delta+1)^2|N_{i-1}|$ vertices belong to the one-hop neighbourhood of $N_{i-1}$ or can reach it through fans. 
By Lemma~\ref{lma:reachingLemma} at most $2(\Delta+1)^4|N_{i-1}|$ points can reach it through the bichromatic part of their Vizing chain. Summing these two contributions counts the first kind. As for the second part, we note that at most $2\Delta|N_{i-1}|$ vertices in $G$ neighbour the endpoints of the bichromatic paths used to extend the Vizing chains at $N_{i-1}$.

Note secondly that at most 
\[
2(\Delta+1)^3 (\Delta+1) \sum \limits_{k = 0}^{i-2} |R_k| + (\Delta+1)^2 \sum \limits_{k = 0}^{i-2} |R_k| \leq 3(\Delta+1)^4\sum \limits_{k = 0}^{i-2} |R_k|
\]
Vizing chain extensions of $R_{i}$ can be overlapping through the one-hop neighbourhood of $\bigcup \limits_{k = 0}^{i-2} R_k$. 
Indeed, as discussed earlier the number of such points that can reach the one-hop neighbourhood of $\bigcup \limits_{k = 0}^{i-2} R_k$ via the bichromatic path component of a Vizing chain is upper bounded by the number of points in $R_i$ that neighbours an endpoint of a maximal bichromatic path that intersects the one-hop neighbourhood of $\bigcup \limits_{k = 0}^{i-2} R_k$ under the initial colouring $c$. 
Since the one-hop neighbourhood has size at most $(\Delta+1) \sum \limits_{k = 0}^{i-2} |R_k|$, at most $(\Delta+1)^2 \sum \limits_{k = 0}^{i-2} |R_k|$ vertices either belong to the one-hop neighbourhood or can reach the one-hop neighbourhood via fans, and at most $2(\Delta+1)^3 (\Delta+1) \sum \limits_{k = 0}^{i-2} |R_k|$ vertices can reach the one-hop neighbourhood via the bichromatic path components of the Vizing chains by Lemma~\ref{lma:reachingLemma}. 
Summing these contributions yields the above. 
Note that since we remove all points that can reach $\bigcup \limits_{k = 0}^{i-2} R_k$ through fans, we are actually removing the two-hop neighbourhood of $\bigcup \limits_{k = 0}^{i-2} R_k$.

Hence, in total at most 
\[
3(\Delta+1)^4 \sum \limits_{k = 0}^{i-2} |R_k|+(3(\Delta+1)^4+2\Delta)|N_{i-1}|
\]
vertices of $R_{i}$ does not belong to $N_{i}$.
\end{proof}
Now we may deduce Lemma~\ref{lma:mainApp2}.
\begin{proof}
By Lemma~\ref{lma:boundN}, we have that
\[
|N_{i}| \geq |R_{i}|-\paren{3(\Delta+1)^4 \sum \limits_{k = 0}^{i-2} |R_k|+(3(\Delta+1)^4+2\Delta)|N_{i-1}|}
\]
We will show that $|N_{i}| \geq |R_{i}|/2$, and then the lemma will follow by Lemma~\ref{lma:boundR}. Observe that by assumptions $1)$ and $2)$ we have
\begin{align*}
    3(\Delta+1)^4 \sum \limits_{k = 0}^{i-2} |R_k|+(3(\Delta+1)^4+2\Delta)|N_{i-1}|
    &\leq 6(\Delta+1)^4 \sum \limits_{k = 0}^{i-2} |N_k|+5(\Delta+1)^4|N_{i-1}| \\
    &\leq \frac{48(\Delta+1)^7}{\ell}|N_{i-1}| +5(\Delta+1)^4|N_{i-1}| \\
    &\leq \paren{\frac{48(\Delta+1)^7}{\ell}+5(\Delta+1)^4}|N_{i-1}|
\end{align*}
Since $\ell \geq 80(\Delta+1)^7$, we find that
\[
\paren{\frac{48(\Delta+1)^7}{\ell}+5(\Delta+1)^4}|N_{i-1}| \leq 6(\Delta+1)^4|N_{i-1}|
\]
and by Lemma~\ref{lma:boundR} that 
\begin{align*}
|R_{i}| &\geq (\frac{\ell}{2\Delta^3}-1)|N_{i-1}|  \\ 
&\geq (40(\Delta+1)^{4}-1)|N_{i-1}| \\
&\geq 20(\Delta+1)^{4}|N_{i-1}| \\
& \geq 2\cdot{} \paren{3(\Delta+1)^4 \sum \limits_{k = 0}^{i-2} |R_k|+(3(\Delta+1)^4+2\Delta)|N_{i-1}|}
\end{align*}
and so we conclude by Lemma~\ref{lma:boundN} that 
\[
|N_{i}| \geq |R_{i}|/2
\]
and by Lemma~\ref{lma:boundR} that 
\[
|N_{i}| \geq \frac{\ell}{8\Delta^3}\sum\limits_{k=1}^{i-1} |N_k|
\]
which was wanted we to show.
\end{proof}
We may now deduce Theorem~\ref{thm:LOCALVT}:
\begin{proof}[Proof of Theorem~\ref{thm:LOCALVT}]
We first construct an augmenting Vizing chain on $e$ as in the proof of Vizing's Theorem. If it has length $ \leq \ell $, we are done, so we may assume this is not the case. 
Now we may choose to truncate the Vizing chain at any of the $\ell$ first vertices, and extend it via Lemma~\ref{lma:fan2} and the discussion following the lemma. At most $2\Delta$ of these chains will overlap the the edges of the path of the first Vizing chain -- namely if they use the exact same bichromatic path. Furthermore, since at most $(\Delta+1)^2$ augmenting paths can go through a vertex, at most $2(\Delta+1)^4$ of the Vizing chains will go through the fan situated at $e$, and so we find that $|R_{1}(e,\ell)| = \ell$ and that $|N_{1}(e,\ell)| \geq \ell-2\Delta-2(\Delta+1)^4 \geq \ell/2$. 
Hence, we may apply Lemma~\ref{lma:mainApp2} until we find a short augmenting and non-overlapping Vizing chain. If this happens at step $i$, then by construction the Vizing chain has size at most $O(\Delta^7 \cdot{} i)$. 
\begin{claim}
For all $j$ satisfying the conditions in Lemma~\ref{lma:mainApp2}, we have $|N_{j}| \geq \paren{\frac{\ell}{8\Delta^3}}^{j}$.
\end{claim}
\begin{proof}
The proof is by induction on $j$. We have just handled the base case $j = 1$ above. The induction step now follows from %combining Lemma~\ref{lma:boundR} i.e.\ that $|R_{j}| \geq \frac{\ell}{4\Delta^{3}}|N_{j-1}|$ and 
condition 1) of Lemma~\ref{lma:mainApp2}, since it allows us to conclude that $|N_{j}| \geq \paren{\frac{\ell}{8\Delta^3}}|N_{j-1}|$ .
\end{proof}
and so we conclude that $i \leq \log_{\frac{\ell}{8\Delta^3}} n$ and the theorem follows. Indeed, the construction has to terminate with a desired Vizing chain after at most $\log_{\frac{\ell}{8\Delta^3}} n$ steps, or we reach a contradiction by above.
\end{proof}

\section{Application 3: Distributed $\Delta +1$ edge-colouring} \label{sec:app3}

We say a family $\mathcal{F} = \{\mathbb{F}_{i}\}_{i = 1}^{\infty}$ is \emph{$k$-bounded}, if for all $i$ we have $|\mathbb{F}_{i}| \leq k^{i}$.  For a $k$-bounded family with $k \geq 2$, it holds that $|\bigcup \limits_{j = 1}^{i} \mathbb{F}_{j} | \leq 2k^{i}$. Indeed, in the worst case the sets in $\mathcal{F}$ are pairwise disjoint and hence 
\[
|\bigcup \limits_{j = 1}^{i} |\mathbb{F}_{j}| | \leq k^{i} \cdot \sum \limits_{j = 0}^{\infty} k^{-j} = \frac{k}{k-1} k^{i} \leq 2k^{i}
\]
We will show that we can construct small augmenting non-overlapping Vizing chains that furthermore also avoid any $4(\Delta+1)^4$-bounded family $\mathcal{F}$. The idea is simply to remove the option to extend the Vizing chains through points forbidden by $\mathcal{F}$. The number of new starting points for Vizing chains will still grow so fast that the number of points to avoid becomes negligible. Hence, we now fix any $4(\Delta+1)^4$-bounded family $\mathcal{F}$ -- even one chosen adversarial -- for the remainder of this section. Then, equivalently to Definitions~\ref{def:Ni} and~\ref{def:Ri}, we have the following definitions: 
\begin{definition} \label{lma:RiF}
Given a proper partial colouring and an uncoloured edge $e$, we let $v \in R_i(e,\ell,\mathcal{F})$ denote the set of vertices such that there exists an $\mathcal{F}$-avoiding and non-overlapping $i$-step Vizing chain from $e$ ending at $v$ such that every Vizing chain in the $i$-step Vizing chain has length at most $\ell$.
\end{definition}
\begin{definition} \label{def:NiF}
Given a proper partial colouring and an uncoloured edge $e$, we let $v \in N_i(e,\ell,\mathcal{F}) \subset R_i(e,\ell,\mathcal{F})$ if there exists an $\mathcal{F}$-avoiding and non-overlapping $i$-step chain from $e$ ending at $v$, and, furthermore, the multi-step Vizing chain may be extended to a non-overlapping $(i+1)$-Vizing chain by concatenating a Vizing chain of length at most $\ell$ at $v$.
\end{definition}
Again if $e,\ell, \mathcal{F}$ are clear from the context, we will often times suppress these arguments. We will prove the following lemma, similar to Lemma~\ref{lma:mainApp2}. 
\begin{lemma} \label{lma:mainApp3}
Suppose we are given a proper partial colouring and an uncoloured edge $e$ in a graph of maximum degree $\Delta$, and let $\ell \geq 100(\Delta+1)^7$.
Suppose that for all $1 \leq j \leq i-1$ it holds that no vertex in $N_j(e,\ell,\mathcal{F})$ can be extended to an augmenting, $\mathcal{F}$-avoiding and non-overlapping $(j+1)$-step Vizing chain via a Vizing chain of length at most $\ell$. Suppose furthermore that the conditions below hold:
\begin{enumerate}
      \item $|N_{j}(e,\ell,\mathcal{F})| \geq \frac{\ell}{8\Delta^3}\sum\limits_{k=1}^{j-1} |N_k(e,\ell,\mathcal{F})|$
      \item $|N_j(e,\ell,\mathcal{F})| \geq |R_j(e,\ell,\mathcal{F})|/2$.
\end{enumerate}
Then condition $1)$ and $2)$ also hold for $j = i$. 
\end{lemma}
Completely synchronous to before, we first show that if the conditions of Lemma~\ref{lma:mainApp3} are satisfied then $R_i(e,\ell,\mathcal{F})$ is big. Since we already removed points forbidden by $\mathcal{F}$ from $N_i$ above, the proof of the following Lemma is (almost) verbatim that of Lemma~\ref{lma:boundR}. We will repeat it below only for completeness:
\begin{lemma} \label{lma:boundRF}
Suppose the conditions of Lemma~\ref{lma:mainApp3} hold. Then 
\[
|R_i| \geq \paren{\frac{\ell}{2\Delta^3}-1}|N_{i-1}| \geq \frac{1}{4} \paren{\frac{\ell}{\Delta^3}} \sum\limits_{k=1}^{i-1} |N_k|
\]
\end{lemma}
\begin{proof}
By assumption, we may extend exactly one $(i-1)$-step Vizing chain through every vertex in $N_{i-1}$ to get non-overlapping $\mathcal{F}$-avoiding $i$-step Vizing chains, where the length of the $i$th Vizing chain is at least $\ell$.
Consider the subgraph $H$ spanned by the $\ell$ first edges of each bichromatic path in the $i$th Vizing chains concatenated at the vertices in $N_{i-1}$. By definition of $R_i$, we have
\[
|V(H)| \leq  |N_{i-1}| +  |R_{i}|
\]
Furthermore, since we only construct one Vizing chain from every point in $N_{i-1}$, an edge in $H$ belongs to at most $4\Delta^2$ different Vizing chains by Lemma~\ref{lma:packingLemma}. 
Since we have assumed that each Vizing chain has length at least $\ell$, each Vizing chain will contribute exactly $\ell$ edges to $H$, and so  we find:
\[
4\Delta^2 |E(H)| \geq \ell \cdot{} |N_{i-1}|
\]
Hence it follows by Observation~\ref{obs:density} that
\[
\frac{\ell|N_{i-1}|}{4\Delta^2(|N_{i-1}| +  |R_{i}|)} \leq \frac{|E(H)|}{ |V(H)|} \leq \frac{\Delta}{2}
\]
By rearranging, we conclude that
\[
|R_i|  \geq \paren{\frac{\ell}{2\Delta^3}-1}|N_{i-1}|
\]
By assumption, we have $|N_{i-1}| \geq \frac{\ell}{8\Delta^3}\sum \limits_{k = 1}^{i-2} |N_{k}|$, so since $\ell \geq 100\Delta^3$ is large enough we find
\[
|N_{i-1}| \geq 10 \sum \limits_{k=1}^{i-2} |N_k|
\]
so 
\[
\sum \limits_{k=1}^{i-1} |N_k| \leq \frac{|N_{i-1}|}{10}+|N_{i-1}|
\]
and it follows that
\[
\frac{10}{11} \sum \limits_{k=1}^{i-1} |N_k| \leq |N_{i-1}|
\]
Hence we find
\[
|R_i| \geq \paren{\frac{\ell}{2\Delta^3}-1}|N_{i-1}| \geq \frac{10}{11} \paren{\frac{\ell}{2\Delta^3}-1}\sum\limits_{k=1}^{i-1} |N_k| \geq \frac{1}{4} \paren{\frac{\ell}{\Delta^3}} \sum\limits_{k=1}^{i-1} |N_k|
\]
\end{proof}
Next we conclude that this implies that $N_{i}$ also is big. To do so, we show that only a constant fraction of the vertices in $R_i(e,\ell)$ are forbidden by $\mathcal{F}$ or are extended by overlapping Vizing chains.
\begin{lemma} \label{lma:boundNF}
Let the assumptions be as in Lemma~\ref{lma:mainApp3}. Then at most
\[
3(\Delta+1)^4 \sum \limits_{k = 0}^{i-2} |R_k|+(3(\Delta+1)^4+2\Delta+4(\Delta+1)^{4})|N_{i-1}|
\]
points in $R_{i}$ have no non-overlapping $\mathcal{F}$-avoiding Vizing chain extension.
\end{lemma}
\begin{proof}
Observe first that if we can guarantee that the Vizing chain beginning at a point in $R_{i}$ will not intersect a vertex in the one-hop neighbourhood of $\bigcup \limits_{k = 0}^{i-1} R_k$, then surely it is non-overlapping. We will show something a little weaker: namely that we can pick the Vizing chains so that most will be vertex disjoint from the one-hop neighbourhood of $\bigcup \limits_{k = 0}^{i-2} R_k$, and furthermore they will not overlap the last Vizing chain used to reach them. 
We may simply remove the 1-hop neighbourhood of vertices in $N_{i-1}$, since we do not wish to truncate our Vizing chains at these points, as we cannot guarantee that the fans constructed at these points are non-overlapping. 

First of all, since we always use Lemma~\ref{lma:fan2} to extend a multi-step Vizing chain, we can count the number of vertices that will extend their Vizing chain so that it overlaps the last Vizing chain used to reach it. 
This can only happen in two ways: either the Vizing chain goes through the fan situated at a vertex in $N_{i-1}$ or it attempts to pick exactly the same augmenting path used by the last Vizing chain to reach it. There are at most $3(\Delta+1)^4 |N_{i-1}|$ chains of the first kind, and at most $2\Delta|N_{i-1}|$ of the second kind. 
Indeed, the one-hop neighbourhood of $N_{i-1}$ has size at most $(\Delta+1)|N_{i-1}|$, and hence by a similar argument to above at most $(\Delta+1)^2|N_{i-1}|$ vertices belong to the one-hop neighbourhood of $N_{i-1}$ or can reach it through fans. 
By Lemma~\ref{lma:reachingLemma} at most $2(\Delta+1)^4|N_{i-1}|$ points can reach it through the bichromatic part of their Vizing chain. Summing these two contributions counts the first kind. As for the second part, we note that at most $2\Delta|N_{i-1}|$ vertices in $G$ neighbour the endpoints of the bichromatic paths used to extend the Vizing chains at $N_{i-1}$.

Note secondly that at most 
\[
2(\Delta+1)^3 (\Delta+1) \sum \limits_{k = 0}^{i-2} |R_k| + (\Delta+1)^2 \sum \limits_{k = 0}^{i-2} |R_k| \leq 3(\Delta+1)^4\sum \limits_{k = 0}^{i-2} |R_k|
\]
Vizing chain extensions of $R_{i}$ can be overlapping through the one-hop neighbourhood of $\bigcup \limits_{k = 0}^{i-2} R_k$.
Indeed, as discussed earlier the number of such points that can reach the one-hop neighbourhood of $\bigcup \limits_{k = 0}^{i-2} R_k$ via the bichromatic path component of a Vizing chain is upper bounded by the number of points in $R_i$ that neighbours an endpoint of a maximal bichromatic path that intersects the one-hop neighbourhood of $\bigcup \limits_{k = 0}^{i-2} R_k$ under the initial colouring $c$. 
Since the one-hop neighbourhood has size at most $(\Delta+1) \sum \limits_{k = 0}^{i-2} |R_k|$, at most $(\Delta+1)^2 \sum \limits_{k = 0}^{i-2} |R_k|$ vertices either belong to the one-hop neighbourhood or can reach the one-hop neighbourhood via fans, and at most $2(\Delta+1)^3 (\Delta+1) \sum \limits_{k = 0}^{i-2} |R_k|$ vertices can reach the one-hop neighbourhood via the bichromatic path components of the Vizing chains by Lemma~\ref{lma:reachingLemma}. 
Summing these contributions yields the above. 
Note that since we remove all points that can reach $\bigcup \limits_{k = 0}^{i-2} R_k$ through fans, we are actually removing the two-hop neighbourhood of $\bigcup \limits_{k = 0}^{i-2} R_k$.

Hence, in total at most 
\[
3(\Delta+1)^4 \sum \limits_{k = 0}^{i-2} |R_k|+(3(\Delta+1)^4+2\Delta)|N_{i-1}|
\]
vertices of $R_{i}$ does not belong to $N_{i}$.

Finally, we have to remove $\bigcup \limits_{j = 1}^{i} \mathbb{F}_{j}$ from $N_{i}$. Earlier, we observed that $|\bigcup \limits_{j = 1}^{i} \mathbb{F}_{j}| \leq 2(4(\Delta+1)^{4})^{i}$.
Hence, in total at most 
\[
3(\Delta+1)^4 \sum \limits_{k = 0}^{i-2} |R_k|+(3(\Delta+1)^4+2\Delta)|N_{i-1}|+2(4(\Delta+1)^{4})^{i}
\]
vertices of $R_{i}$ does not belong to $N_{i}$. The lemma will follow from the following claim:
\begin{claim} \label{clm:boundNiF}
For all $j$ satisfying the conditions in Lemma~\ref{lma:mainApp3}, we have $|N_{j}| \geq \paren{\frac{\ell}{8\Delta^3}}^{j} \geq 2(4(\Delta+1)^{4})^{j}$.
\end{claim}
\begin{proof}
The proof is by induction on $j$. For $j = 1$, we have $|R_{1}| = \ell$. At most $2\Delta$ of the chains obtained by extending the Vizing chain through a vertex in $R_1$ will overlap the edges in the path of the first Vizing chain -- namely if they choose the exactly same bichromatic path. Furthermore, at most $3(\Delta+1)^4$ of the chains will go through the fan situated at $e$ by Lemma~\ref{lma:reachingLemma}, and so we find that $|N_{1}| \geq \ell-2\Delta-3(\Delta+1)^4 - |\mathbb{F}_{1}| \geq \ell/2 \geq \frac{\ell}{8\Delta^3} \geq 2(4(\Delta+1)^4)$. The induction step now follows by condition 1) in Lemma~\ref{lma:mainApp3}, and noting that $\frac{\ell}{8\Delta^3} \geq 4(\Delta+1)^{4}$
\end{proof}
Now $4(\Delta+1)^{4}|N_{i-1}| \geq 2(4(\Delta+1)^{4})^{i}$ and the lemma follows.
\end{proof}
Now we may deduce Lemma~\ref{lma:mainApp3}.
\begin{proof}[Proof of Lemma\ref{lma:mainApp3}]
By Lemma~\ref{lma:boundN}, we have that
\[
|N_{i}| \geq |R_{i}|-\paren{3(\Delta+1)^4 \sum \limits_{k = 0}^{i-2} |R_k|+(3(\Delta+1)^4+2\Delta+4\Delta^3)|N_{i-1}|}
\]
We will show that $|N_{i}| \geq |R_{i}|/2$, and then the lemma will follow by Lemma~\ref{lma:boundRF}. Observe that by assumptions $1)$ and $2)$ we have
\begin{align*}
    3(\Delta+1)^4 \sum \limits_{k = 0}^{i-2} |R_k|+(3(\Delta+1)^4+2\Delta+4(\Delta+1)^4)|N_{i-1}|
    &\leq 6(\Delta+1)^4 \sum \limits_{k = 0}^{i-2} |N_k|+9(\Delta+1)^4|N_{i-1}| \\
    &\leq \frac{48(\Delta+1)^7}{\ell}|N_{i-1}| +9(\Delta+1)^4|N_{i-1}| \\
    &\leq \paren{\frac{48(\Delta+1)^7}{\ell}+9(\Delta+1)^4}|N_{i-1}|
\end{align*}
Since $\ell \geq 100(\Delta+1)^7$, we find that
\[
\paren{\frac{48(\Delta+1)^7}{\ell}+9(\Delta+1)^4}|N_{i-1}| \leq 10(\Delta+1)^4|N_{i-1}|
\]
and by Lemma~\ref{lma:boundRF} that 
\begin{align*}
|R_{i}| &\geq (\frac{\ell}{2\Delta^3}-1)|N_{i-1}|  \\ 
&\geq (50(\Delta+1)^{4}-1)|N_{i-1}| \\
&\geq 25(\Delta+1)^{4}|N_{i-1}| \\
& \geq 2\cdot{} \paren{3(\Delta+1)^4 \sum \limits_{k = 0}^{i-2} |R_k|+(3(\Delta+1)^4+2\Delta+4\Delta^3)|N_{i-1}|}
\end{align*}
and so we conclude by Lemma~\ref{lma:boundNF} that 
\[
|N_{i}| \geq |R_{i}|/2
\]
and by Lemma~\ref{lma:boundRF} that 
\[
|N_{i}| \geq \frac{\ell}{8\Delta^3}\sum\limits_{k=1}^{i-1} |N_k(e,\ell)|
\]
which was what we wanted to show.
\end{proof}
Note that by Claim~\ref{clm:boundNiF}, we will find an augmenting, $\mathcal{F}$-avoiding and non-overlapping Vizing chain after at most $\log_{\frac{\ell}{8\Delta^3}} n$ steps. We may now use this to prove Theorem~\ref{thm:algoLOCALVT}. 
The idea is that augmenting an edge only invalidates a limited amount of points for other vertices to extend their Vizing chains through. In particular, this implies that only few edges will have to avoid families that are not $4(\Delta+1)^4$-bounded, and so we can extend the set of edges which may be augmented in parallel. 
\begin{theorem} \label{thm:paraAug}
For any partial colouring $c$ of a graph $G$, let $U$ be the set of edges left uncoloured by $c$. Then there exists a set $W \subset U$ of size $\Omega\paren{\frac{|U|}{ \poly(\Delta) \log n}}$ of edges for which there exist vertex disjoint augmenting Vizing chains, each of size at most $O(\Delta^7 \log n)$.
\end{theorem}
\begin{proof}
As noted above, the idea is that whenever we have constructed $k$ disjoint Vizing chains, we can use these Vizing chains to define families to avoid for each remaining uncoloured edges. 
Since an uncoloured edge can avoid families that are $4 (\Delta+1)^3$-bounded, but each vertex in a Vizing chain only invalidates $2(\Delta+1)^3$ points, even if the graph is laid out in the worst possible way, every vertex in a chosen Vizing chain invalidates at most $O(\Delta^5)$ other edges from joining $W$. 

Let us be more specific. Suppose we have a set $W \subset U$ of size $k$ such that there exist pairwise disjoint augmenting Vizing chains of size at most $(\Delta+1) \cdot{} \ell \cdot{} \log n$ for each edge in $W$. Fix such a set of $k$ pairwise disjoint Vizing chains and let $G' \subset G$ be the subgraph containing exactly these Vizing chains.

We show that if $k < \frac{|U|}{12poly(\Delta) \log n}$, then we may extend $W$ to contain $k + 1$ edges, while satisfying the same conditions. 

Before we construct these families, we first consider any edge $e \in U - W$. 
If any edge in $G'$ has an endpoint in the 1-hop neighbourhood of $e$, then we will completely disregard $e$. Otherwise, we consider the points that are 1) not in the $1$-hop neighbourhood in $G$ of any of the augmenting Vizing chains in $G'$ and 2) that can reach a vertex in $N^{1}(G')$ via a Vizing chain. 
By Lemma~\ref{lma:reachingLemma} there are at most 
$$2(\Delta+1)^3\cdot{} (\Delta+1)|V(G')|+(\Delta+1)^2|V(G')| \leq (3(\Delta+1)^4)^2|V(G')| $$
such points. Here, the first contribution comes from the vertices that reach $N^{1}(G')$ via the path chains, and the second contribution from the vertices that reach $N^{1}(G')$ via the fan chains.  
A similar argument may be applied to the points reaching these point via $1$-step Vizing chains, which shows that at most $(3(\Delta+1)^4)^2|V(G')|$ points can reach the one-hop neighbourhood of $G'$ via non-overlapping $2$-step Vizing chains. 
Following this logic, a straightforward inductive argument shows that at most $(3(\Delta+1)^{4})^{i}|V(G')|$ points can reach them via non-overlapping $i$-step Vizing chains. 
% Note that after the i^{\textrm{th}} step, we also have to count the 0 step fans that could be used to reach the first points. Hence again we multiply by $(\Delta+1)$ and get the $Delta+1$ to the 4. Only $\Delta^{i}$ of things reaching things through continous fan augmentation.
For any $i$, we let $P_{i-1}$ be the multi-set of points which can reach the $1$-hop neighbourhood through non-overlapping $i$-step Vizing chains. We add a point once for each $i$-step Vizing chain, beginnnig at the point, that reaches $G'$ through edges in $G-G'$. Then by above \[
|P_{i-1}| \leq  (3(\Delta+1)^{4})^{i} \cdot{} |V(G')| \leq 3(\Delta+1)^{4}(3(\Delta+1)^{4})^{i-1} \cdot{} k \cdot{} s
\]
where $s = (\Delta+1) \cdot{} \ell \cdot{} \log n$.

We construct a family $\mathcal{F}(e)$ for $e$ to avoid as follows: if a non-overlapping $j$-step Vizing chain constructed on $e$ reaches a vertex $w$ in the 1-hop neighbourhood of $G'$ in $G$, then add the last point we extended through before reaching $w$ to $\mathbb{F}_{j-1}(e)$. Note that in particular this means that any $\mathcal{F}(e)$-avoiding non-overlapping multi-step Vizing chain constructed on $e$ has to be vertex disjoint from the one-hop neighbourhood of $G'$, since we never extend through a point, which can reach said neighbourhood through a $1$-step Vizing chain.
We say that the family for $e$ is \emph{unavoidable} if it is not $4(\Delta+1)^4$-bounded. Note that one might actually be able to avoid it, but for the sake of this analysis we will simply assume that this is not the case.

A vertex $w$ has to be represented more than $(4(\Delta+1)^4)^{i}$ times in $P_i$ for some $i$ for the family of that vertex to be unavoidable.
Indeed, for any edge $e$ incident to $w$, we have $|\mathbb{F}_{i}(e)| \leq$ the number of times $w$ is present in $P_i$. 
This is because we add $w$ to $P_i$ once for each $(i+1)$-step Vizing chain beginning at $w$ that reaches $G'$, and we add at most one vertex to $|\mathbb{F}_{i}(e)|$ for each such Vizing chain.
Hence, the maximum number of \emph{unavoidable} families that we construct is upper bounded by:
\begin{align*}
    (\Delta+1)k\cdot{} s + \sum \limits_{i = 1}^{\log_{\frac{\ell}{4\Delta^3}} n} \frac{|P_{i}|}{(4(\Delta+1)^{4})^{i}} & \leq (\Delta+1)k\cdot{} s + 3(\Delta+1)^{4}\sum \limits_{i = 1}^{\infty} \frac{k (3(\Delta+1)^{4})^{i} \cdot{} s}{(4(\Delta+1)^{4})^{i}} \\
    & \leq 3(\Delta+1)^{4}\cdot{}k\cdot{} s\sum \limits_{i = 0}^{\infty} \paren{\frac{3}{4}}^{i} \\
    & \leq 12(\Delta+1)^{4}\cdot{}k\cdot{} s
\end{align*}
hence if $k < \frac{|U|}{12(\Delta+1)^{4}s}$, then there exists an edge $e \in U$ such that $e$ has an augmenting non-overlapping Vizing chain of size $s$ that is completely disjoint from $G'$, and so we may extend $W$ by adding this edge. Indeed, since $e$ avoids $\mathcal{F}(e)$, none of the Vizing chains considered are allowed to be extended through a point that can reach $G'$ by a $1$-step Vizing chain. This means that any subgraph created in this manner will be vertex disjoint from $G'$. 
\end{proof}
In particular, this theorem now allows us to prove Theorem~\ref{thm:algoLOCALVT}, which we restate here for the readers' convenience. 
\begin{theorem}[Theorem~\ref{thm:algoLOCALVT}]
There is a deterministic LOCAL algorithm that computes a $(\Delta+1)$-edge-colouring in $\tilde{O}(\poly(\Delta)\log^6 n)$ rounds.
\end{theorem}
The approach from now on is completely symmetrical to that of Bernshteyn~\cite{BERNSHTEYN}. We apply the following Theorem due to Harris:
\begin{theorem}[Harris, Theorem 1.1~\cite{harris2019distributed}] \label{thm:harris}
There is a deterministic LOCAL algorithm that outputs an $O(r)$ approximation to a maximum matching on a hypergraph on $n$ vertices with maximum degree $d$ and rank $r$ in 
\[
\tilde{O}(r \log d + \log^2 d + \log^* n)
\]
rounds. Here $\tilde{O}(x)$ hides $\poly(\log x)$ factors.  
\end{theorem}
The idea is to define a hypergraph for which a matching corresponds to disjoint augmenting Vizing chains that we may augment in parallel. Theorem~\ref{thm:paraAug} will show that this hypergraph contains a 'large' matching if many edges are left uncoloured.
Define a hypergraph $H$ on the vertex set of $G$ as follows: for each augmenting, non-overlapping Vizing chain of an uncoloured edge, where each step of the Vizing chain has length at most $\ell$ and each Vizing chain has at most $\log n$ steps, we add a hyper-edge containing exactly the vertices of the Vizing chain. 
For every uncoloured edge, we have at most $\ell^{\log n}$ such chains, so the maximum degree $d$ of a vertex in this graph is upper bounded by $ d \leq m\cdot{}\ell^{\log n}$. 

By Theorem~\ref{thm:paraAug}, the size of a maximum matching is at least $\Omega \paren{\frac{|U|}{\poly(\Delta)\log n}}$. Finally, the maximum \emph{rank} of a hyper-edge is $O(\ell \cdot{} \log n)$ by construction, and so Theorem~\ref{thm:harris} finds a matching of size at least $\Omega\paren{\frac{|U|}{\poly(\Delta)\log^2 n}}$.

By applying Theorem~\ref{thm:harris}, we may in $\tilde{O}(\poly(\Delta)\log^2 n)$ rounds on the hyper graph $H$ construct an approximate maximum matching in the hypergraph. Since it takes $O(\poly(\Delta)\log n)$ rounds on $G$ to simulate each round on the hypergraph $H$, we find that we can find such a maximum matching in $\tilde{O}(\poly(\Delta)\log^3 n)$ rounds in $G$. Once the maximum matching is found, we augment the corresponding Vizing chains in parallel in $O(\poly(\Delta)\log n)$ rounds.

Hence, it is enough to give an upper bound on the number of times this procedure needs to be run. Each time it is run, the number of uncoloured edges $|U|$ drops from $|U|$ to $|U|\paren{1-\Omega \paren{\frac{1}{\poly(\Delta)\log^2 n}}}$. So we need to perform at most 
\begin{align*}
    \log_{\paren{\frac{1}{1-\Omega \paren{\frac{1}{\poly(\Delta)\log^2 n}}}}} n &= \frac{\log n}{\log\paren{\frac{1}{1-\Omega \paren{\frac{1}{\poly(\Delta)\log^2 n}}}}} \\
    & \leq \frac{\log n}{\log \paren{1+\Omega \paren{\frac{1}{\poly(\Delta)\log^2 n}}}} \\
    & \leq \frac{\log n}{\Omega \paren{\frac{1}{\poly(\Delta)\log^2 n}}} \\
    & \leq O(\poly(\Delta)\log^3 n)
\end{align*}
where we used that for small enough $x$ we have $\frac{x}{2} \leq \frac{x}{1+x} \leq \log(1+x)$. Thus we need to simulate at most $O(\poly(\Delta)\log^3)$ rounds on $H$, resulting in a total round complexity of $\tilde{O}(\poly(\Delta) \log^6 n)$. Combining this gives Theorem~\ref{thm:algoLOCALVT}.

\section{Application 4: Dynamic$(1+\varepsilon)\Delta$-edge-colouring} \label{sec:dynCol}
The overall structure of our approach is similar to that of Duan, He and Zhang~\cite{duan}. Let us first recap their high-level approach. Their approach has 3 steps. 
First they make an amortised reduction to the easier non-adaptive version of the problem, where we may assume $\Delta$ is fixed. The idea is to use multiple copies of the graph where each copy has a different upper bound on the maximum degree, so that one can apply a non-adaptive algorithm on each copy. 
Finally, they describe how to move points between these copies as their degrees change. The key part that makes the reduction amortised is that one might have to reinsert a point into a copy with a smaller degree-cap, when the point has its degree lowered. 
It is not immediately clear whether such an approach can be de-amortised, since we may not be able to insert edges incident to a high-degree vertex prematurely. 

The second step is to reduce the problem to one where $\Delta = O(\frac{\log n}{\varepsilon})$. This is done by sampling a subset $S \subset [\Delta + 1]$ such that every vertex in $G$ has an available colour in $S$. Such a subset is called a \emph{palette}. Given a palette $S$, one can restrict oneself to $G[S]$ i.e.\ the subgraph spanned by all edges coloured with a colour from $S$, since $S$ was chosen exactly so as to make the standard Vizing chain constructions go through. The third and final step is to show how to insert an uncoloured edge into such a low-degree graph. To do this, Duan, He and Zhang~\cite{duan} show that one can construct non-overlapping multi-step Vizing chains by using disjoint palettes for each step, and that short Vizing chains of this type have a good probability of being augmenting. 
Since they need to sample disjoint palettes for each step, each vertex need to have many free colours, and so one needs a lower bound on the maximum degree of roughly $\Delta = \Omega(\frac{\log^2 n}{\varepsilon^2})$ for this approach to work. Note that deletions are trivial to handle at this point, since $\Delta$ is fixed and we just need to produce a $(\Delta + 1)$-edge-colouring. Note also that Duan, He and Zhang~\cite{duan} in fact constructed non-overlapping Vizing chains.

Our approach will circumvent the first step, but be similar to the last two steps. The starting point is the observation that shifting an augmenting Vizing chain only increases the number of edges coloured with colours present at the bichromatic path of the Vizing chain. These colours are chosen as colours that are available at some of the vertices in the fan. Hence, if these vertices all have low-degree it seems unnecessary to choose available colours which are far larger than the degrees of these vertices and in that way colour a lot of edges using colours $\kappa$ where $\kappa$ is large. Motivated by this observation, we show that one can sample a form for local palettes, where each vertex $v$ only picks available colours in $[(1+\varepsilon)d(v)]$. We then show that if one constructs Vizing chains using such a local palette, then we can maintain the invariant that the number of edges coloured $\kappa$ is upper bounded by the number of vertices $v$ for which $\kappa \in [(1+\varepsilon)d(v)]$ and $v$ is also incident to an edge coloured $\kappa$. Furthermore, we show that we only have to recolour $O(1)$ edges to restore this invariant after a deletion. This allows us to reduce the problem to the case where $\Delta$ is assumed to be small. 
Luckily in this case we can apply a lemma proved by Bernshteyn~\cite{BERNSHTEYN} to -- with good probability -- construct short augmenting multi-step Vizing chains, and use these Vizing chains to colour any uncoloured edges. Since the Vizing chains can be constructed without sampling multiple palettes, this also allows us to remove the condition on $\Delta$.

\subsection{Local Palettes}
As noted earlier a key problem with sampling a palette from $[(1+\varepsilon)\Delta]$ is that we may end up choosing available colours from the upper end of the interval to colour edges incident to vertices of very low degree. 
This is a problem since if $\Delta$ drops too far, we may have to recolour all of these edges. If there are many such edges this is time consuming. 
Instead, we will sample a \emph{local palette} in which each vertex $v$ will pick an available colour from $[(1+\varepsilon)d(v)]$. 
Note that to be completely precise, we would have to choose a colour in $[\lceil (1+\varepsilon)d(v) \rceil]$, but for ease of exposition we will just write $[(1+\varepsilon)d(v)]$. 
At all times it is straightforward to replace $[(1+\varepsilon)d(v)]$ by $[\lceil (1+\varepsilon)d(v) \rceil]$.
Recall that $A_{\varepsilon}(v) = A(v) \cap [(1+\varepsilon)d(v)]$. We have:
\begin{definition} \label{def:localPalette}
A subset of colours $S \subset [(1+\varepsilon)d(v)]$ is said to be an \emph{$\varepsilon$-local palette} if it holds for all $v \in V(G)$ that $|S \cap A_{\varepsilon}(v)| \geq 1$ i.e.\ if every vertex $v$ has an available colour in $S$ that is also in $[(1+\varepsilon)d(v)]$.
\end{definition}
We will show that there exist small local palettes, and that we may sample one efficiently. 
Before showing this, we show that we may chop $[(1+\varepsilon)\Delta]$ into smaller intervals such that we can sample from these intervals to get our local palette.
To this end note that given $0 <\varepsilon \leq 1$ and $\Delta$, we may partition $[(1+\varepsilon)\Delta]$ into intervals $I_{i}$ such that $I_{i} = [2^{i},2^{i+1}]$ for $i \leq \lceil \log{(1+\varepsilon)\Delta} \rceil$. 
We will say a vertex $v$ is \emph{$\varepsilon$-dense} in an interval $I_{i}$ if $|A_{\varepsilon}(v) \cap I_{i}| \geq \varepsilon |I_{i}|/4$. The following lemma says that all vertices are $\varepsilon$-dense in at least one meaningful interval.
\begin{lemma} \label{lma:epsDense}
For any vertex $v$ there exists an $i$ such that $v$ is $\varepsilon$-dense in $I_i$.
\end{lemma}
\begin{proof}
Consider first the interval that $(1+\varepsilon)d(v) \in I_{j}$ belongs to. Now, if $v$ is not $\varepsilon$-dense in $I_j$, it certainly holds that $v$ has less than $d(v)\varepsilon/2$ available colours in $I_{j}$. Indeed, if $2^{j} > (1+\varepsilon/2)d(v)$ this holds trivially, so we may assume that $2^{j} \leq d(v)(1+\varepsilon/2)$. This means that $|I_{j}| \leq (1+\varepsilon/2)d(v)$ and so if $v$ has $d(v)\varepsilon/2$ available colours in $I_j$, $v$ would be \emph{$\varepsilon$-dense} in $I_{j}$. 
Indeed, it holds that
\[
\frac{\varepsilon d(v)}{2} \geq \frac{\varepsilon(1+\varepsilon)d(v)}{4} \geq |I_{j}| \cdot{} \frac{\varepsilon}{4}
\]

Suppose furthermore that no $i$ exists such that $v$ is $\varepsilon$-dense in $I_i$. Then it holds that
\begin{align*}
    |A_{\varepsilon}(v)| &< \varepsilon d(v)/2 + \sum \limits_{k = 1}^{j-1} \frac{\varepsilon}{4} \cdot{} 2^k
                         \leq \varepsilon d(v)/2 + \frac{\varepsilon}{4} \cdot{} 2^{j} \\
                         &\leq \varepsilon d(v)/2 + \frac{\varepsilon}{4} \cdot{} (1+\varepsilon)d(v) 
                         \leq \varepsilon d(v)
\end{align*}
Contradicting the fact that $v$ has at least $\varepsilon d(v)$ free colours. 
\end{proof}
We will use this lemma to show how to sample a local palette for each vertex $v$. 
Our strategy will be to sample a sufficient amount of points inside each interval $I_{i}$. By Lemma~\ref{lma:epsDense}, we know that every vertex is \emph{$\varepsilon$-dense} in some meaningful interval, and so when we sample this interval, we will hit an available colour for said vertex with high probability. We show this formally in the following lemma.
\begin{lemma} \label{lma:localPalSize}
For any $\varepsilon, \Delta$ there exists an $\varepsilon$-local palette of size $O(\frac{\log n \log \Delta}{\varepsilon})$. Furthermore, we can construct such a palette in $O(\frac{\log n \log \Delta}{\varepsilon})$ time with high probability. 
\end{lemma}
\begin{proof}
We will sample $S$ as follows: for each $j \in [\lceil \log (1+\varepsilon)\Delta \rceil]$ if $|I_{j}| \leq \frac{c\log n}{\varepsilon}$, we set $S_j = I_{j}$. Otherwise, we sample $\frac{c\log n}{\varepsilon}$ points uniformly at random from $I_j$, and let the sampled points form $S_j$. Clearly $S$ has the claimed size. We will show that $S = \bigcup \limits_{k} S_{k}$ is a local palette with probability $1-n^{c/4-1}$.

Consider an arbitrary vertex $v$. We will upper bound the probability that $S$ is not an $\varepsilon$-local palette for $v$. By Lemma~\ref{lma:epsDense}, we know that $v$ is $\varepsilon$-dense in some $I_j$. Now if $|I_j| \leq \frac{c\ln n}{\varepsilon}$, we include the entire interval in $S$, and so certainly $S$ is $\varepsilon$-local for $v$. 
Otherwise, $v$ has at least $\varepsilon \cdot{} |I_j| / 4$ free colours. Hence, the probability that none are sampled may be upper bounded by $(1-\varepsilon/4)^{c \frac{\ln n}{\varepsilon}} = (1-\frac{c\ln n}{4} \cdot{} \frac{\varepsilon}{c \ln n})^{c \frac{\ln n}{\varepsilon}} \leq e^{-\frac{c}{4} \ln n} = n^{-c/4}$.
Hence union bounding over all choices of $v$ shows that $S$ is a local palette with probability at least $1-n^{c/4-1}$.
\end{proof}

\subsection{Colouring Protocol} \label{subsec:col}
In this section, we briefly discuss how to replace the construction of non-overlapping augmenting Vizing chains applied by Duan, He and Zhang~\cite{duan} with the construction of Bernshteyn~\cite{BERNSHTEYN} in order to get rid of the assumption that $\Delta  = \Omega(\frac{\log^2 n}{\epsilon^2})$. Then in the next section, we will show that using this type of construction together with a local palette, we can adapt to changing $\Delta$. 

Consider the following multi-step Vizing algorithm due to Bernshteyn:
fix an uncoloured edge $e$. Now construct an augmenting Vizing chain on $e$ using the standard approach from Vizing's Theorem. If the Vizing chain has length $ \ell = \Omega(\Delta^6 \log n)$, pick an edge uniformly at random among the first $\ell$ edges of the bichromatic path and truncate the Vizing chain here. Now recursively construct a multi-step Vizing chain by applying Lemma~\ref{lma:fan2} and the discussion immediately following the lemma. If the concatenated Vizing chain is both short and augmenting, we stop and augment it. Otherwise if it becomes overlapping, we stop. If the Vizing chain is still not short and augmenting after $\Omega(\log n)$ steps, we  also stop. Bernshteyn showed the following lemma:
\begin{lemma}[Bernshteyn, Lemma 6.1 p. 346 in~\cite{BERNSHTEYN}] \label{lma:BER}
For any $T \in \mathbb{N}$ and any $\ell\in \mathbb{N}$ such that $\lambda := \frac{\ell}{(\Delta+1)^3} > 1$, then the multi-step Vizing chain algorithm described above terminates with an augmenting, non-overlapping $i$-step Vizing chain for some $i < T$ with probability at least
\[
1-\frac{n}{\lambda^T}-\frac{3T(\Delta + 1)^3}{\lambda - 1}
\]
\end{lemma}
We will first use Lemma~\ref{lma:localPalSize} to sample an $\varepsilon$-local palette $S$. Then we will attempt to colour the uncoloured edge by restricting ourselves to $G[S]$, and then use the algorithm of Bernshteyn to attempt to construct an augmenting, $\varepsilon$-local and non-overlapping Vizing chain. 
To make the Vizing chains $\varepsilon$-local, we choose only $\varepsilon$-available colours as representative available colours, unless Lemma~\ref{lma:fan2} case $2)$ forces us to pick a colour $\kappa_1$ that is not $\varepsilon$-available at the center of the fan. However, as we show in the next section, the constructed Vizing chain is still $\varepsilon$-local.
Apart from this the constructions are similar to what we have seen before, but for completeness we elaborate on them in Appendix~\ref{app:B}. 
It is straightforward to adapt the proofs of Bernshteyn to accommodate these changes, so we will not elaborate further on this.
If the Vizing chain construction fails, we will resample $S$ and try again. See Algorithm~\ref{alg:Algo1} for pseudo-code.
\begin{algorithm} 
    \caption{An algorithm for colouring uncoloured edges.}
    \label{alg:Algo1}
    \begin{algorithmic}
            \Function{colour}{$e = uv$}
                \State Sample a local palette $S$. 
                \State Construct an $\varepsilon$-local Vizing chain $C_1$ on $uv$ in $G[S]$. 
                \State $i = 1$ 
                \While{$length(C_i) > \ell$}
                    \State Pick an edge $e'$ u.a.r. from first $\ell$ edges on $C_i$ 
                    \State $i = i + 1$ 
                    \State Truncate Vizing chain at $e'$ and extend it via the $\varepsilon$-local version of Lemma~\ref{lma:fan2}.
                    \If{Multi-step Vizing chain is overlapping or $i > T$} 
                        \State $\texttt{colour}(uv)$ 
                        \State Stop
                    \EndIf 
                \EndWhile
                \State Augment $F_1 + C_1 + \dots + F_{i} + C_{i}$ 
            \EndFunction
        \end{algorithmic}
\end{algorithm}
Thus we have the following:
\begin{corollary} \label{cor:Bern}
If $T = 2\log n$ and $\ell = 1+18(\Delta+1)^6 \log n$, then Algorithm~\ref{alg:Algo1} is called recursively at most ($c+1) \log n$ times with probability $1-n^{-c}$.
\end{corollary}
\begin{proof}
Assume first that all the sampled colour-sets form $\varepsilon$-local palettes. Then by Lemma~\ref{lma:BER} a call to $\operatorname{colour}(uv)$ is successful with probability at least 
\[
1-\frac{n}{(\Delta^3\log n)^{2 \log n}}-\frac{6(\Delta+1)^3 \log n}{18(\Delta + 1)^3 \log n} \geq 1-\frac{1}{n}-\frac{1}{3} \geq \frac{1}{2}
\]
Hence a call to $\operatorname{colour}(uv)$ fails with probability at most a half. Hence the probability that it fails $(c+1) \log n$ times in a row is at most $n^{-(c+1)}$. Finally, by sampling large enough palettes, it follows by Lemma~\ref{lma:localPalSize} that each colour-set is an $\varepsilon$-local palette with probability at least $n^{-(c+1)}$. Union-bounding over all of these events yields the result.
\end{proof}
Note that the following lemma then is immediate from Corollary~\ref{cor:Bern} and Lemma~\ref{lma:localPalSize} through the use of basic data structures. 
\begin{lemma} \label{lma:Algo2}
Using $\operatorname{colour}(uv)$ an edge can be inserted in time $O(\varepsilon^{-6} \log ^6 \Delta \log^9 n)$ with high probability.
\end{lemma}
\begin{proof}
By Lemma~\ref{lma:localPalSize}, we can sample an $\varepsilon$-local palette in $O(\frac{\log \Delta \log n}{\varepsilon})$ time. We can also use it to construct fans in $O(\frac{\log \Delta \log n}{\varepsilon})$ time per vertex in a fan. We may chase bichromatic paths easily by storing edges in an array indexed by their colour at each vertex. Hence, by Corollary~\ref{cor:Bern} we arrive at the claimed running time. 
\end{proof}

\subsection{Reducing deletions to insertions} \label{sec:RedDel}
In this section, we show how to handle deletions while staying adaptive to the maximum degree $\Delta$. To illustrate the basic idea, consider first the incremental setting, where we only insert edges. 
Suppose, furthermore, that all edges are inserted via Algorithm~\ref{alg:Algo1} from the previous section using local palettes. 
The key observation is then that the number of edges coloured with the colour $\kappa$ is upper bounded by $|\{v: \kappa \in [(1+\varepsilon)d(v)], \kappa \notin A_{\varepsilon}(v)\}|$.
In fact, we can say something stronger: namely that if this invariant holds before an edge is coloured, then it also holds after, and so we conclude that the Vizing chains are in fact $\varepsilon$-local. 
This allows us to use a simple recolouring scheme to restore the invariant, whenever an edge is deleted. 

We first show that the Vizing chains constructed does not invalidate the invariant below, provided that the invariant was true before the shift. That is we show that the constructed Vizing chains are in fact $\varepsilon$-local.
\begin{invariant} \label{inv:boundC}
For all $\kappa$ we have $|\{e:c(e) = \kappa \}| \leq |\{v:  \kappa \in [(1+\varepsilon)d(v)], \kappa \notin A_{\varepsilon}(v)\}|$.
\end{invariant}
\begin{lemma} \label{lma:invMaintainingSingle}
Suppose that Invariant~\ref{inv:boundC} holds for all $\kappa$. Then it also holds after shifting a Vizing chain $F+P$ constructed as above (i.e.\ as in Appendix~\ref{app:B}).
\end{lemma}
\begin{proof}
Consider first the fan $F$ on edges $uw_1, \dots, uw_k$. Clearly $|\{e:c(e) = \kappa \}|$ remains unchanged for all choices of $\kappa$ after shifting $F$, and by construction $|\{v:  \kappa \in [(1+\varepsilon)d(v)], \kappa \notin A_{\varepsilon}(v)\}|$ cannot decrease. Indeed, this follows from arguments similar to those in the proof of Lemma~\ref{lma:mainLVT}. 
Namely, $u$'s contribution to this set is unchanged after the fan is shifted. 
For each colour $\kappa$, at most one vertex from $\{v:  \kappa \in [(1+\varepsilon)d(v)], \kappa \notin A_{\varepsilon}(v)\}$ is removed during the fan shift, but one vertex is also sure to be added by construction.

We first consider the case where $F+P$ is truncated. Similarly to Lemma~\ref{lma:mainLVT}, we note that only the contributions of the center of the fan $u$, the first vertex of $P$, $x_0 = w_k$, and the last two vertices of $P$, $x_{p-1}$ and $x_p$ are changed when $P$ is shifted. 

Observe that $|\{e:c(e) = \kappa \}|$ remains unchanged. Indeed, suppose $P$ is $(\kappa_1,\kappa_2)$-bichromatic with $\kappa_1$ available at $u$ and $\kappa_2$ available at $x_0$. If we truncate at an edge coloured $\kappa_1$, we will have equally many edges coloured $\kappa_1$ and $\kappa_2$ before, and so this lost edge is made up for by $e$ being coloured $\kappa_1$. If we instead cut at an edge coloured $\kappa_2$, then we will have exactly one more edge before that is coloured $\kappa_1$ than edges coloured $\kappa_2$, and so $|\{e:c(e) = \kappa_{i} \}|$ remains unchanged for $i \in \{1,2\}$ in this case also.
Now observe that $u$ will contribute to $|\{v:  \kappa_1 \in [(1+\varepsilon)d(v)], \kappa_1 \notin A_{\varepsilon}(v)\}|$ and $x_0$ will contribute to $|\{v:  \kappa_2 \in [(1+\varepsilon)d(v)], \kappa_2 \notin A_{\varepsilon}(v)\}|$ after the shift of the path. Furthermore, in the worst case these contributions are cancelled out by the lost contributions of $x_{p-1}$ and $x_p$. % Do we need to handle degenerate cases? 

In the case where $P$ is augmenting, we note that if $\kappa_1$ was available at $u$, then $u$ will still contribute to $|\{v:  \kappa_1 \in [(1+\varepsilon)d(v)], \kappa_1 \notin A_{\varepsilon}(v)\}|$ and $x_0$ will still contribute to $|\{v:  \kappa_2 \in [(1+\varepsilon)d(v)], \kappa_2 \notin A_{\varepsilon}(v)\}|$ after the shift of the path.
Now, however, in the worst case $x_p$ stops contributing to $|\{v:  \kappa_i \in [(1+\varepsilon)d(v)], \kappa_i \notin A_{\varepsilon}(v)\}|$ for some $i$. Either way if $i = 2$, this does not matter, as then $|\{e:c(e) = \kappa_2\}|$ remains unchanged and the lost contribution is balanced out by $x_0$'s new contribution. Otherwise, if $i = 1$, it still does not matter as then $|\{e:c(e) = \kappa_1\}|$ is unchanged and the lost contribution is balanced out by $u$'s new contribution. 
In other words if $x_p$ stops contributing to $|\{v:  \kappa_1 \in [(1+\varepsilon)d(v)], \kappa_1 \notin A_{\varepsilon}(v)\}|$ then it is $|\{e:c(e) = \kappa_2\}|$ that increases and vice versa. 

We remark here that we do not have to consider the case where $P$ ends back at the fan, since the algorithm due to Bernshteyn rules out this type of augmenting chains.
\end{proof}
Of course the above only holds for a $1$-step Vizing chain, but it is straightforward to extend the arguments to hold for a multi-step Vizing chain:
\begin{lemma} \label{lma:invMaintaining}
Suppose that Invariant~\ref{inv:boundC} holds for all $\kappa$. Then it also holds after shifting a multi-step Vizing chain $F_{1}+P_{1} + \dots + F_{i} + P_{i}$ constructed as above.
\end{lemma}
\begin{proof}
The proof will be by induction. 
The only thing that we have to be careful of, compared to Lemma~\ref{lma:invMaintainingSingle}, is a Vizing chain extended via condition $2)$ in Lemma~\ref{lma:fan2}. 
Here, $\kappa_1$ might not be $\varepsilon$-available at the fan, and so when we analyse the impact of shifting this fan, we are not certain that the center of the fan will contribute to $|\{v:  \kappa_1 \in [(1+\varepsilon)d(v)], \kappa_1 \notin A_{\varepsilon}(v)\}|$, 
This is of course not a problem, since this means that the vertex also was not contributing to this quantity beforehand, and so the previous chain that was shifted can account for this offset.

To be able to handle the above formally, we define an \emph{$\varepsilon$-damaged} subchain to be a maximal subchain $F_j+P_j + \dots F_{j+r}+P_{j+r}$ such that for all $0\leq h \leq r$ $P_{j+r}$ is $(\kappa_1, \kappa_2)$-bichromatic and $\kappa_1 \in A_{\varepsilon}(u_{j})$, but the representative available colour at $u_{j+h}$ is not in $[(1+\varepsilon)d(u_{j+h})]$ for $0<h\leq r$. 

Notice that any subchain $F_j + P_j$ is part of a unique \emph{$\varepsilon$-damaged} subchain, and so we can uniquely decompose our $i$-step Vizing chain into \emph{$\varepsilon$-damaged} subchains. 
In particular, if we can show that any \emph{$\varepsilon$-damaged} subchain is $\varepsilon$-local, then it follows by induction on the number of \emph{$\varepsilon$-damaged} subchains in the decomposition of the $i$-step Vizing chain that the entire $i$-step Vizing chain is $\varepsilon$-local.
\begin{claim}
Let $F_j+P_j + \dots F_{j+r}+P_{j+r}$ be an \emph{$\varepsilon$-damaged} subchain. 
Then $F_j+P_j + \dots F_{j+r}+P_{j+r}$ is an $\varepsilon$-local Vizing chain.
\end{claim}
\begin{proof}
Let $c$ be the pre-shift colouring of the \emph{$\varepsilon$-damaged} subchain. 
We note first that any colour that is not $\kappa_1$ or $\kappa_2$ can only be shifted in fans, which we already saw in the proof of Lemma~\ref{lma:invMaintainingSingle} does not invalidate the invariant for these colours. 
Hence, we only have to focus on $\kappa_1$ and $\kappa_2$.
We will prove the following by induction: after shifting 
$F_j+P_j + \dots F_{j+h}+P_{j+h}$ for $h < r$, we have that 
\[
|\{e:c(e) = \kappa(h+1) \}|<|\{v:  \kappa(h+1) \in [(1+\varepsilon)d(v)], \kappa(h+1) \notin A_{\varepsilon}(v)\}|
\]
and 
\[
|\{e:c(e) = \kappa'(h+1) \}|\leq |\{v:  \kappa'(h+1) \in [(1+\varepsilon)d(v)], \kappa'(h+1) \notin A_{\varepsilon}(v)\}|
\]
where $\kappa(h) \in \{\kappa_1, \kappa_2\}$ is the representative available colour at $u_{j+h}$, and $\kappa'(h) \in \{\kappa_1, \kappa_2\} \setminus \kappa(h)$  is the other colour. 
This provides enough offset to also shift $F_{j+r}+P_{j+r}$ without breaking the invariant. 

The proof is by induction on $h$. We let $u_h$ and $w_{h,k}$ be the center of the fan $F_{j+h}$ respectively the first vertex of $P_{j+h}$. 
We also let $x_{p-1,h}$ and $x_{p,h}$ be the second last respectively last vertex on $P_{j+h}$.
The base case is $h = 0$. 
Recall from the proof of Lemma~\ref{lma:invMaintainingSingle} that for all $\kappa$ $|\{e:c(e) = \kappa \}|$ remains unchanged.
Now $u_h$ will contribute to $|\{v:  \kappa(h) \in [(1+\varepsilon)d(v)], \kappa(h) \notin A_{\varepsilon}(v)\}|$ and $w_{h,k}$ will contribute to $|\{v:  \kappa'(h) \in [(1+\varepsilon)d(v)], \kappa'(h) \notin A_{\varepsilon}(v)\}|$ after the shift of the path by the assumption that $\kappa_1 = \kappa(0) \in A_{\varepsilon}(u)$. In the worst case the contribution to $|\{v:  \kappa'(h+1) \in [(1+\varepsilon)d(v)], \kappa'(h+1) \notin A_{\varepsilon}(v)\}|$ is cancelled out by the lost contribution of $x_{p,h}$.
But since, by assumption, $\kappa(h+1) \notin [(1+\varepsilon)d(u_{h+1})]$, there is no lost contribution from $x_{p-1,h} = u_{h+1}$, and so we now have 
\[
|\{e:c(e) = \kappa(h+1) \}|<|\{v:  \kappa(h+1) \in [(1+\varepsilon)d(v)], \kappa(h+1) \notin A_{\varepsilon}(v)\}|
\]
We proceed to the induction step. Assume the statement holds for $h$. We show it also holds for $h+1 < r$. 
There are two cases: either $\kappa(h) = \kappa(h+1)$ or $\kappa(h) = \kappa'(h+1)$. \\\\
$\mathbf{Case}\text{ } \kappa(h) = \kappa(h+1):$ in this case, we know that that shifting $F_{j+h+1}+P_{j+h+1}$ only increases $|\{v:  \kappa'(h) \in [(1+\varepsilon)d(v)], \kappa'(h) \notin A_{\varepsilon}(v)\}|$ by 1, which, as always, in the worst case is cancelled out by the lost contribution of $x_{p,h+1}$.
But we also know $x_{p-1,h+1} = u_{h+1}$ does not have $\kappa(h+1)$ as an $\varepsilon$-available colour after the shift, and so it must be that $\kappa(h+1) \notin [(1+\varepsilon)d(u_{h+1})]$. 
Hence, as
\[
|\{e:c(e) = \kappa(h+1) \}|<|\{v:  \kappa(h+1) \in [(1+\varepsilon)d(v)], \kappa(h+1) \notin A_{\varepsilon}(v)\}|
\]
was true before the shift, it also holds after the shift. \\\\
$\mathbf{Case}\text{ } \kappa(h) \neq \kappa(h+1):$ in this case, we know that that shifting $F_{j+h+1}+P_{j+h+1}$ only increases $|\{v:  \kappa'(h) \in [(1+\varepsilon)d(v)], \kappa'(h) \notin A_{\varepsilon}(v)\}|$ by 1. 
In the worst case $x_{p,h+1}$ stops contributing to $\kappa'(h+1) = \kappa(h)$, and so $|\{v:  \kappa'(h+1) \in [(1+\varepsilon)d(v)], \kappa'(h+1) \notin A_{\varepsilon}(v)\}|$ might drop by 1. But by assumption, since $\kappa'(h+1) = \kappa(h)$, we had 
\[
|\{e:c(e) = \kappa'(h+1) \}|<|\{v:  \kappa'(h+1) \in [(1+\varepsilon)d(v)], \kappa'(h+1) \notin A_{\varepsilon}(v)\}|
\]
before the shift, and so we still have 
\[
|\{e:c(e) = \kappa'(h+1) \}|\leq |\{v:  \kappa'(h+1) \in [(1+\varepsilon)d(v)], \kappa'(h+1) \notin A_{\varepsilon}(v)\}|
\]
after the shift. We also know $x_{p-1,h+1} = u_{h+1}$ does not have $\kappa(h+1)$ as an $\varepsilon$-available colour after the shift, and so it must be that $\kappa(h+1) \notin [(1+\varepsilon)d(u_{h+1})]$. 
Hence, we find that
\[
|\{e:c(e) = \kappa(h+1) \}|<|\{v:  \kappa(h+1) \in [(1+\varepsilon)d(v)], \kappa(h+1) \notin A_{\varepsilon}(v)\}|
\]
after the shift. 

This proves the statement. The claim now follows by showing that shifting $F_{j+r}+P_{j+r}$ does not invalidate the $\varepsilon$-locality of the subchain. 
To that end, we note that shifting $F_{j+r}+P_{j+r}$ increases the contribution from $w_{k,r}$ to $|\{v:  \kappa'(h+1) \in [(1+\varepsilon)d(v)], \kappa'(h+1) \notin A_{\varepsilon}(v)\}|$ by 1. 
The lost contributions from $x_{p-1,r}$ and $x_{p,r}$ in the worst case drops $|\{v:  \kappa(h+1) \in [(1+\varepsilon)d(v)], \kappa(h+1) \notin A_{\varepsilon}(v)\}|$ and $|\{v:  \kappa'(h+1) \in [(1+\varepsilon)d(v)], \kappa'(h+1) \notin A_{\varepsilon}(v)\}|$ by 1. 
This cancels out the above new contribution for $\kappa'(r)$, and since we had 
\[
|\{e:c(e) = \kappa(r) \}|<|\{v:  \kappa(r) \in [(1+\varepsilon)d(v)], \kappa(r) \notin A_{\varepsilon}(v)\}|
\]
before the shift, we still have 
\[
|\{e:c(e) = \kappa(r) \}|\leq |\{v:  \kappa(r) \in [(1+\varepsilon)d(v)], \kappa(r) \notin A_{\varepsilon}(v)\}|
\]
after. Hence the claim follows.
\end{proof}
As mentioned earlier, the lemma now follows by induction on the number $\varepsilon$-damaged subchains.
\end{proof}
Lemma~\ref{lma:invMaintaining} says that if we make sure to recolour edges in such a way that we adhere to Invariant~\ref{inv:boundC} for all $\kappa$ during deletions, then we can also colour inserted edges in such a way that we adhere to Invariant~\ref{inv:boundC} for all $\kappa$. There are only two things to consider, when maintaining the invariant under deletions. First of all, when deleting an edge $xy$ coloured $\kappa$, one of the following things happen: neither $x$ nor $y$ contribute to $|\{v:  \kappa \in [(1+\varepsilon)d(v)], \kappa \notin A_{\varepsilon}(v)\}|$, exactly one of $x$ and $y$ contribute to $|\{v:  \kappa \in [(1+\varepsilon)d(v)], \kappa \notin A_{\varepsilon}(v)\}|$ or both $x$ and $y$ contribute to $|\{v:  \kappa \in [(1+\varepsilon)d(v)], \kappa \notin A_{\varepsilon}(v)\}|$. If neither $x$ nor $y$ or exactly one of them contribute to $|\{v:  \kappa \in [(1+\varepsilon)d(v)], \kappa \notin A_{\varepsilon}(v)\}|$, then all is fine. A problem only arises if both $x$ and $y$ contribute to $|\{v:  \kappa \in [(1+\varepsilon)d(v)], \kappa \notin A_{\varepsilon}(v)\}|$ -- in this case we will say the edge is a \emph{$\kappa$-heavy} edge. Then if the invariant held with equality beforehand, one would now possibly have to recolour an edge coloured $\kappa$ where neither of the endpoints contribute to $|\{v:  \kappa \in [(1+\varepsilon)d(v)], \kappa \notin A_{\varepsilon}(v)\}|$ in order to restore Invariant~\ref{inv:boundC}. Call such an edge \emph{$\kappa$-light}. Since the inequality held with equality before the deletion, it then follows that such a \emph{$\kappa$-light} edge has to exist, and so we can just uncolour it in order to restore the invariant. Hence, we have reduced the problem to that of colouring this edge, and we already now how to do so while maintaining the invariant. 

Second of all the degree of $u$ and $v$ might drop. This could also invalidate the invariant for some $\kappa'$. However, this can only happen for $O(1)$ choices of $\kappa'$. If the invariant fails to hold because $u$ or $v$ dropped their degree, it is because at least one of them (possibly both) are incident to a $\kappa'$-coloured edge and so they stop contributing to $|\{v:  \kappa \in [(1+\varepsilon)d(v)], \kappa \notin A_{\varepsilon}(v)\}|$. In this case, however we can just recolour any such edges. Since there are at most $O(1)$ such edges, we have reduced the case of handling deletions to the case of colouring uncoloured edges, which we already know how to do. In order to be able to locate $\kappa$-light edges efficiently, we will with $O(1)$ overhead make sure that we update a doubly-linked list $Q_{\kappa}$ for each colour $\kappa$ containing all $\kappa$-light edges. Such lists are straightforwardly maintained under degree increments and decrements as well as under Vizing chain shifts. 
Indeed, we can calculate the $O(1)$ affected colours for each vertex and locate and fix any potentially affected edges in constant time. 
In summary, we have shown that we can maintain Invariant~\ref{inv:boundC} during deletions using the simple Algorithm~\ref{alg:deletions}: 
\begin{algorithm} 
    \caption{Deletions}
    \label{alg:deletions}
    \begin{algorithmic}
            \Function{Delete}{$e = uv$}
                \State Initialise $R$ as a stack.
                \If{$uv$ is $c(uv)$-heavy}
                    \State $e' = $ an arbitrary edge from $Q_{c(uv)}$
                    \State Remove $e'$ from $Q_{c(uv)}$
                    \State Uncolour $e'$
                    \State Push $e'$ to $R$
                \EndIf
                \For{$x \in \{u,v\}$}
                    \For{all $\kappa \in ((1+\varepsilon)(d(x)-1),(1+\varepsilon)d(x)])$}
                        \State Let $e'$ be an edge coloured $\kappa$ incident to $x$.
                        \State Uncolour $e'$
                        \State Push $e'$ to $R$.
                    \EndFor
                \EndFor
                \State Delete $e$ from the graph.
                \While{$R$ non-empty}
                    \State $e' = \texttt{pop}(R)$ 
                    \State $\texttt{colour}(e')$
                \EndWhile
            \EndFunction
        \end{algorithmic}
\end{algorithm}
We have essentially shown the following lemma.
\begin{lemma} \label{lma:delete}
$\operatorname{Delete}$ runs in $O(\varepsilon^{-6} \log ^6 \Delta \log^9 n)$ time, and if Invariant~\ref{inv:boundC} holds before $\operatorname{Delete}$ is run, then it also holds after. 
\end{lemma}
\begin{proof}
Correctness follows from the above discussion which shows that Invariant~\ref{inv:boundC} is restored before the $\operatorname{while}$-loop is entered. Lemma~\ref{lma:invMaintaining} shows that it also holds after the while-loop has been exited. 
As for the running time, it is completely straightforward to maintain such doubly linked-lists for every colour with only $O(1)$ space and time overhead, provided that we store a pointer at each edge to its position in the relevant doubly-linked lists. 
\end{proof}
Hence, we have shown Theorem~\ref{thm:dynVT}, which we restate for the readers' convenience:
\begin{theorem}[Theorem~\ref{thm:dynVT}]
Let $G$ be a dynamic graph subject to insertions and deletions with current maximum degree $\Delta$. There exists a fully-dynamic and $\Delta$-adaptive algorithm maintaining a proper $(1+\varepsilon)\Delta$-edge-colouring with $O(\varepsilon^{-6} \log ^6 \Delta \log^9 n)$ update time with high probability.
\end{theorem}

\section*{Acknowledgements}
The author would like to thank Eva Rotenberg for helpful discussions and encouragements. 

\bibliographystyle{alpha}
\bibliography{refs}

\appendix

\section{Constructing $\varepsilon$-local Vizing chains} \label{app:B}
For completeness, we outline exactly how to construct an $\varepsilon$-local Vizing chain in this appendix. 

\paragraph{Construction:} Let $c$ be a partial colouring of a graph $G$, and let $e = uv \in E(G)$ be any edge left uncoloured by $c$. 
Then we consider the following Vizing chain on $e$:
we construct a maximal fan $uw_1, \dots, uw_k$ as follows. Initially, we set $w_1 = v$. Now having picked $w_j$, we pick $w_{j+1}$ or conclude our construction with $k = j$ as follows: Note first that for all vertices $y$ we have that $A(y) \cap [\lceil (1+\varepsilon)d(y) \rceil]$ is non-empty, since at most $d(y)$ colours can be excluded from $A(y)$.
We now have three cases: \\\\
\textbf{Case 1:} if there exists a colour $\kappa \in A(w_j) \cap [\lceil(1+\varepsilon)d(w_j)\rceil]$ that is also available at $u$, we pick $\kappa$ as $w_j$'s representative available colour and conclude the construction with $k = j$. \\\\
\textbf{Case 2:} if there exists a colour $\kappa \in A(w_j) \cap [\lceil(1+\varepsilon)d(w_j)\rceil]$ and an index $\ell<j$ such that $c(uw_{\ell}) = \kappa$, we pick $\kappa$ as $w_j$'s representative available colour and conclude the construction with $k = j$. \\\\
\textbf{Case 3:} Otherwise, we pick an arbitrary colour $\kappa \in A(w_j) \cap [\lceil(1+\varepsilon)d(w_j)\rceil]$ as $w_j$'s representative available colour. Note that $u$ has to be incident to an edge $ux$ coloured $\kappa$ and that $x \notin \{w_{\ell}\}_{\ell=1}^{j}$, since then we would be in case $1)$ or case $2)$. Therefore, we can pick $w_{j+1} = x$. \\\\
The above construction gives us a fan $F$ on the edges $uw_1, \dots, uw_{k}$, such that $w_k$'s representative available colour is either available at $u$ or is present at $u$ at the edge $uw_i$ for some $i<k$. In particular, if we stop in case $1)$, we have identified a an augmenting Vizing chain, and we conclude the construction.

If we stop in case $2)$, we pick an arbitrary representative available colour $\kappa_1 \in A(u) \cap [\lceil (1+\varepsilon)d(y) \rceil]$ for $u$ and set $\kappa_2$ equal to $w_k$'s representative available colour. Then we consider the $(\kappa_1,\kappa_2)$-bichromatic path $P'$ rooted at $w_k$. 
Now we can truncate $P'$ and extend via Lemma~\ref{lma:fan2} as usual as long as we are careful to pick all representative colours in Lemma~\ref{lma:fan2} as colours that are $\varepsilon$-available. 
This is extension is completely synchronous to previous constructions, and so we will not elaborate on it here, although we will note that if one ends up in case $2)$ of Lemma~\ref{lma:fan2}, one is forced to pick $\kappa_1$ as the representative available colour at $u$, and this colour might not be $\varepsilon$-available at $u$. 
However, as we show in Section~\ref{sec:RedDel} the constructed Vizing chain is still $\varepsilon$-local.

\end{document}